\documentclass{article} 
\usepackage[left=1in, right=1in, bottom=1in, top=1in]{geometry}
\usepackage{amsmath}
\usepackage{xcolor}
\usepackage[colorlinks=true, urlcolor=blue,citecolor=blue,anchorcolor=blue]{hyperref}
\usepackage[capitalize]{cleveref}

\usepackage{amssymb, amsfonts, amsthm, graphicx,bm}
\usepackage{mathtools}
\usepackage{mymacros}
\usepackage[normalem]{ulem}

\newcommand{\Cut}{\Delta}
\newcommand{\cut}{\Cut}
\newcommand{\projS}{P_S} 
\newcommand{\projSbar}{P_{\overline{S}}} 
\newcommand{\HS}{H_{S}} 

\newcommand{\gsproj}{P_\lambda} 
\newcommand{\gse}{\lambda} 
\newcommand{\gsprojS}{P_\mu} 
\newcommand{\gseS}{\mu} 
\newcommand{\gsevalSbar}{\overline{\mu}} 
\newcommand{\h}{h} 
\newcommand{\gapS}{\Gamma_S} 
\newcommand{\gsprojSperp}{P_\mu^\perp}
\newcommand{\gsrS}{\lambda_S} 

\newcommand{\id}{I} 
\renewcommand{\vert}{|}

\usepackage{authblk}

\makeatletter
\renewcommand\paragraph{\@startsection{paragraph}{4}{\z@}%
                                    {1ex \@plus1ex \@minus.2ex}%
                                    {-1em}%
                                    {\normalfont\normalsize\bfseries}}
\makeatother

\begin{document}

\title{On the stability of solutions to Schr\"{o}dinger's equation short of the adiabatic limit}
\author[1,2]{Jacob Bringewatt}
\author[1]{T. C. Mooney}
\author[3]{Michael Jarret}
\affil[1]{Joint Center for Quantum Information and Computer Science and Joint Quantum Institute, NIST/University of Maryland College Park, Maryland 20742, USA}
\affil[2]{Department of Physics, Harvard University, Cambridge, MA 02138, USA}
\affil[3]{Department of Mathematical Sciences, Department of Computer Science, and Quantum Science and Engineering Center George Mason University, Fairfax, VA, 22030, USA}

\maketitle

\begin{abstract}
    We prove an adiabatic theorem that applies at timescales short of the typical adiabatic limit. Our proof analyzes the stability of solutions to Schrodinger's equation under perturbation. We directly characterize cross-subspace effects of perturbation, which are typically significantly less than suggested by the perturbation's operator norm. This stability has numerous consequences: we can (1) find timescales where the solution of Schrodinger's equation converges to the ground state of a subspace, (2) lower bound the convergence to the global ground state by demonstrating convergence to some other known quantum state, (3) guarantee faster convergence than the standard adiabatic theorem when the ground state of the perturbed Hamiltonian is close to that of the unperturbed Hamiltonian, and (4) bound leakage effects in terms of the global spectral gap when the Hamiltonian is ``stoquastic'' (a $Z$-matrix). Our results can help explain quantum annealing protocols that exhibit faster convergence than is guaranteed by a standard adiabatic theorem. Our upper and lower bounds demonstrate that at timescales short of the adiabatic limit, subspace dynamics can dominate over global dynamics. Thus, we see that, when our results apply, convergence to particular global target states can be understood as the result of local dynamics. 
\end{abstract}

\section{Introduction}
 Let $A(t) \in \mathbb{C}^{N \times N}$ with $t \in [0,T]$ be a one-parameter family of bounded, Hermitian operators acting on the Hilbert space $\mathcal{H}=\mathbb{C}^N$ endowed with the standard inner product from quantum theory which is
anti-linear in the first argument and linear in the second.
Assume that $A(t)$ has a unique smallest eigenvalue for all $t$ and let the eigenvalues of $A(t)$ be denoted as $\lambda(t)=\lambda_0(t) < \lambda_1(t) \leq \lambda_2(t) \dots \leq \lambda_{N-1}(t)$ with corresponding eigenvectors $\vec\lambda(t),\vec\lambda_1(t),\dots,\vec\lambda_{N-1}(t)$.
A quantum adiabatic theorem is any theorem that bounds the convergence of solutions of Schr\"odinger initial value problems 
\begin{equation}\label{eqn:Schrodinger}
    \begin{cases}
        i \frac{d \vec\psi}{dt} = A(t) \vec\psi(t) \\
        \vec\psi(0) = \vec\lambda_j(0)
    \end{cases}
\end{equation}
to $\vec\lambda_j(t)$ in $[0,T]$ as $T\to \infty$ (Most frequently, $j$ is taken to be $0$, however these theorems can apply to any choice of $j \in \mathbb{Z}_{N}$). The operator $A$ is the Hamiltonian.\footnote{One might object to the fact we are calling this operator $A$ instead of the typical (and natural) choice $H$. However, $H$ will soon be used to refer to a Hamiltonian with a special block-diagonal structure and, consequently, we hold off on using it.} Adiabatic theorems, originally introduced by Einstein \cite{einstein1914verh} and Ehrenfest \cite{Ehrenfest}, were not provided a general, rigorous treatment until Kato in 1950 \cite{Kato}. 
The folklore version of the adiabatic ``theorem'' states (see e.g.~\cite{RolandPRA2002}):
\begin{sthm}\label{sthm:folk_adiabatic}
For the Schr\"odinger initial value problem considered in \cref{eqn:Schrodinger} with $j=0,$ 
\[
{\abs*{\langle \vec\psi(t), \vec\lambda_0(t)\rangle} \geq 1 - \frac{1}{2}\left( \max_{t\in[0,T]}\frac{\norm{\frac{dA}{dt}}}{\left(\lambda_1(t)-\lambda_0(t)\right)^2}\right)^2.}
\]
\end{sthm}
In other words, if $A$ varies slowly enough with respect to $t$, then the solution to Schr\"{o}dinger's equation initialized at $t=0$ to the lowest eigenvector $\vec\lambda_0(0)$ of $A(0)$ remains instantaneously close to the eigenvector $\vec\lambda_0(t)$ at all times $t \in [0,T]$. Hereafter, we will refer to the difference $\gamma(t):=\lambda_1(t)-\lambda_0(t)$ as the \textbf{spectral gap} or \textbf{gap} of $A$.

This folklore theorem and its original rigorous formulations have been expanded upon in two major directions that are relevant for contextualizing this work. First, a variety of adiabatic theorems under different conditions have been proven in the mathematical physics literature. Second, these theorems have been applied to great effect in the context of quantum computing to an analog model of computation known as quantum annealing~\cite{annealingPRE1998, farhi2000quantum,Hauke_2020, AlbashRevModPhys2018}.  Our work straddles the boundaries of these two directions, taking qualitative inspiration from a collection of more nuanced adiabatic theorems to add a new adiabatic theorem to the pantheon that we expect to be of particular use for rigorously understanding phenomena in quantum annealing beyond the standard adiabatic limit. 

The simplest extensions of adiabatic theorems beyond the early results of Kato rely on additional smoothness assumptions on the Hamiltonian to prove tighter and more rigorous bounds that also depend on higher derivatives of $A$ and higher powers of the gap (see e.g. Refs.~\cite{teufel2003adiabatic,jansen2007bounds} and Theorem~\ref{thm:JRS}, below). 
However, there also exist adiabatic theorems that relax the gap assumption, and work with Hamiltonians which have continuous spectra~\cite{bornemann1998homogenization,avron1999adiabatic,salem2007adiabatic,elgart2011adiabatic}.  
Such theorems rely on two key conceptual points: (1) while a state belonging to a continuous spectrum of eigenvalues is not stable to perturbations and, generically, becomes a ``resonance'' that leaks away, an adiabatic theorem can still apply to such states for timescales short of the resonance lifetime; and, further, (2) such results hold even for initial states that are ``nearly'' spectral projections of the Hamiltonian (i.e. only ``nearly'' eigenstates). In the context of these works such ``near'' eigenstates are typically thought of as corresponding to spatially localized wavefunctions. 
In a somewhat related vein, Refs.~\cite{panati2002space,panti2003space} also consider the problem of  applying an adiabatic theorem to a local region of physical space with small corrections to the adiabatic theorem bounded by the effects of tunneling out of this region. Such results have recently been extended to the many-body setting~\cite{teufel2020non,yin2023pretherm}.
Also worth noting is Ref.~\cite{mozgunov2023quantum} which derives a rigorous adiabatic theorem for unbounded Hamiltonians based on applying a cutoff between a low and high energy subspace.

In this work, making use of some qualitatively similar features to these earlier papers, we prove an adiabatic theorem that applies to a mathematical, rather than a physical, subspace. While physical subspaces can be viewed as a special case of our setting our results apply to a more general class of Hamiltonians. 
Analogous to the case of adiabatic theorems without a gap, our result will apply for timescales short of some other relevant timescale in the system. Thus, our result can be viewed as a sort of ``intermediate timescale adiabatic theorem''. Similar to the spatial adiabatic theorems, the long timescale will be determined by the accumulation of error due to leakage of probability amplitude from one subspace to another.  Furthermore, like the adiabatic theorems without a gap, our results will also depend on showing that our associated adiabatic theorem holds for states ``near'' a relevant eigenstate. Like Ref.~\cite{mozgunov2023quantum}, our result depends on separating dynamics into high and low energy subspaces. But our results also control the interactions between states within the low energy subspace, yielding a distinct interpretation and set of applications.

In particular, one of our goals is to provide a rigorous understanding of quantum annealing beyond the adiabatic regime. In quantum annealing, one first initializes the system (typically, a collection of $n$ qubits) in the easy-to-prepare ground state of a Hamiltonian $H_0$, and then performs a time-dependent evolution to some other Hamiltonian $H_1$, whose ground state encodes the solution to a computational problem of interest. The goal of such a procedure is to produce a final state with strong support on the ground state of $H_1$. Typically, this ground state is designed to encode the solution to a computational problem, such as a combinatorical optimization problem.
In the adiabatic regime, success is guaranteed with high probability, provided the Hamiltonian of the system is changed ``slowly enough,'' where slowly enough is rigorously defined via an adiabatic theorem. However, adiabatic evolution is only a sufficient condition for a successful quantum annealing protocol. Consequently, there has been significant work on designing successful quantum annealing protocols with a faster time-to-solution---this includes ideas such as the quantum approximate optimization algorithm (QAOA)~\cite{farhi2014quantum}, counterdiabatic driving~\cite{demirplak2003adiabatic,berry2009transitionless,odelin2019shortcuts}, and optimal control~\cite{odelin2019shortcuts}. 

Unfortunately, most of this work has been heuristic or numerical in nature. One notable exception is Ref.~\cite{garcia2022lower}, which provides rigorous lower bounds on quantum annealing times via quantum speed limits; however, while these bounds were shown to yield an achievable asymptotic scaling for a handful of toy models with large amounts of symmetry, the generic nature of quantum speed limits means that these bounds are unlikely to help characterize the timescales of more realistic quantum annealing algorithms. 
Our theorem provides a powerful alternative tool for rigorously analyzing non-adiabatic quantum annealing. We expect that in most generic instances where faster-than-adiabatic annealing is successful (e.g. as observed numerically and experimentally in Refs.~\cite{ebadi2022quantum,cain2023quantum,braida2023anti}) corresponds to regimes where the dynamics can be shown to decompose into local subspaces where our localized adiabatic theorem applies. 

\subsection{Our Contributions}

In order to provide a semi-rigorous intuition for the substance of our result---and before delving into the technical details---we provide a slightly hand-wavy ``theorem'' that captures its key elements. To do so, we require a bare minimum of notation. Let $S\subset \mathcal{H}$ be a linear subspace and let $\overline{S}:=S^\perp$ be its orthogonal complement. Given any subspace $Q$, let $P_Q$ denote its associated orthogonal projection. 
Letting $s:=t/T\in[0,1]$, we consider a one-parameter family of Hamiltonians $A(s)=H(s)+\Delta(s)$, where $H(s)$ respects the decomposition $\mathcal{H}=S\oplus \overline{S}$,  { $P_S \vec\lambda(s) \neq 0$} for all $s$, and $\Delta(s)$ is a block-antidiagonal perturbation. That is, for all $s$, $[H(s),P_S]=[H(s),P_{\overline{S}}]=0$, $\Delta(s)P_S=P_{\overline{S}}\Delta(s)$, and $\Delta(s)P_{\overline{S}}=P_{{S}}\Delta(s)$. For future convenience, we also define
 \begin{equation}\label{eq:HSHSbar}
 H_S(s):=P_SH(s)P_S \qquad \text{and} \qquad
 H_{\overline{S}}(s):=P_{\overline{S}}H(s)P_{\overline{S}}.
 \end{equation}

Subject to some minor additional assumptions on $H$ and $S$, we prove that something like the standard adiabatic theorem applies to the state $\vec\psi(t)$ at timescales lower bounded by the local spectral gap (of $H_S$) and upper bounded by a weighted version of the Cheeger ratio $h$~\cite{chung2000weighted,jarret2018hamiltonian} (see \cref{def:cheeger} below), which can be understood to quantify the rate of transitions of the state $\vec\psi(t)$ out of $S$ under evolution by the Schr\"odinger equation in \cref{eqn:Schrodinger}. That is, we prove the following ``theorem''  (see \cref{s:results} for the fully rigorous statement). 
\begin{sthm}[Main result]\label{sthm:mainresult}
Let $s:=t/T\in[0,1]$. Let $U(s=1)$ be the unitary evolution generated by the Hamiltonian $H(s)+\Delta(s)$, where $H(s)=H_S+H_{\overline{S}}$ is block diagonal and $\Delta(s)$ is a block anti-diagonal perturbation. Let $V_{\mathrm{ad}}(1)$ be the exact adiabatic unitary evolution generated by $H(s)$. Then,  there exists some $z>0$ such that when acting on the initial ground space ${P_{\mu}(0)}$ of the block $S$
\[
\norm{(U(1)-V_{\mathrm{ad}}(1))
{P_{\mu}(0)}
}
\leq \mathcal{O}\big(\max_{s\in[0,1]}\sqrt{h^{z}T}\big)+\mathcal{O}(\text{usual adiabatic error with Hamiltonian } H_S).
\]
\end{sthm}

This result bounds the error between the ``true'' unitary evolution $U$ generated by the Hamiltonian $H(s)+\Delta(s)$ and adiabatic unitary evolution $V_{\mathrm{ad}}$ generated by the block diagonal Hamiltonian $H(s)$. Note that the bound consists of two pieces: the first term scales with the total time $\sim\sqrt{T}$ and can be intuitively understood as the error that arises due to leakage between the subspaces $S,$ and $\overline{S}$.~\footnote{{{In the quantum annealing literature, this leakage is sometimes referred to as ``tunneling.'' However, this is not \emph{bona fide} tunneling in physical space as $S$ and $\overline{S}$ are subspaces of Hilbert space, not regions of physical space. Thus, we use the term ``leakage.''}}} The rate of this leakage is determined by the Cheeger ratio $h$. The second term scales like $\sim 1/T$ and can be understood as something like the usual adiabatic error applied to the subspace $S$. Importantly, the standard adiabatic error is governed by the eigenvalue gap between the smallest and next-smallest eigenvalues of the relevant Hamiltonian. Similarly, this error is determined by the ``local'' eigenvalue gap of $H_S$, as opposed to the global eigenvalue gap of $H(s)+\Delta(s)$, as in the standard adiabatic theorem. This local gap is generically larger than the global one.

Given the interplay between an error term that scales like $\sqrt{T}$ and one that scales like $1/T$, we observe how our result can be interpreted as an ``intermediate timescale adiabatic theorem,'' where the timescale of applicability is governed by the Cheeger ratio and the local gap.

Our results can be utilized in three primary ways: (1) To prove that, under perturbation, there exist timescales where the solution of Schrodinger's equation converges to the lowest eigenstate of a block, (2) extending this, to provide a lower bound on convergence to the ground state, by demonstrating convergence to some other (known) quantum state, and (3) when the ground state of a perturbed Hamiltonian $H + \Delta$ is close to the ground state of an unperturbed, block-diagonal $H$, to provide faster convergence than that predicted by the standard adiabatic theorem.

More specifically, we emphasize again that we expect these mathematical tools to be particularly applicable to the analysis of quantum annealing beyond the adiabatic regime.
In these cases, we are predominantly interested in well-defined sequences of sets Hamiltonians of ever-increasing Hilbert space dimension $N$, where we wish to understand the rate of convergence with respect to $N$. Although it is possible to analyze the asymptotic behavior in terms of $N$ without reference to spectral information, the difficulty and generality of the computational problems that these sequences correspond to usually leads to applying spectral methods.

The rest of the paper is organized as follows: In \cref{s:mathematicalprelim} we provide some initial definitions and prove a simple proposition relating the spectra of the various Hamiltonians and the Cheeger ratio $h$ that will be used extensively throughout. In \cref{s:timeindependentcase} we discuss a simple version of our theorem for the time independent scenario as a simplified, stage-setting example. In \cref{s:results} we present the formal version of ``Theorem''~\ref{sthm:mainresult} and in \cref{s:details} we provide the technical details that lead to its proof. In \cref{s:stoquasticH} we discuss a simplified version of our main result for the case of so-called ``stoquastic'' Hamiltonians, a sub-class of interest for adiabatic quantum computing. Finally, in \cref{s:examples} we provide a pair of toy examples of how our results can be applied and in \cref{s:discussion} we summarize our conclusions and discuss how we expect our results can be more broadly applied to problems of interest in quantum annealing.

\section{Preliminaries}\label{s:mathematicalprelim}
Here, we develop some additional notation and derive the preliminary results that are needed to prove the rigorous version of ``Theorem'' \ref{sthm:mainresult}. {{We encourage the reader to initially skim this section just enough to follow the time independent case presented in \cref{s:timeindependentcase}, which serves as a mathematically simpler precursor to the more general result.}}

To begin, we provide the formal definition of the Cheeger ratio:
\begin{definition}[Cheeger ratio]\label{def:cheeger}
    For $H, \Cut$ defined as above the \textbf{Cheeger ratio} is defined as
    \[
    \h(s):=\begin{cases}
        -\frac{\langle \vec\lambda,\Cut \projS \vec\lambda\rangle}{\min\{\langle{\vec\lambda},P_S{\vec\lambda}\rangle, \langle{\vec\lambda},\projSbar{\vec\lambda}\rangle\}}, & \text{if } \min\{\langle{\vec\lambda},P_S{\vec\lambda}\rangle, \langle{\vec\lambda},\projSbar{\vec\lambda}\rangle\} \neq 0 \\
        0, & \text{otherwise}.
    \end{cases} 
    \]
\end{definition}
Note that \cref{def:cheeger} is slightly abusive compared to the standard definition for the Cheeger ratio, where one would let $h_S$ be the quantity defined in \cref{def:cheeger} and let $h:=\min_{S} h_S$ be the Cheeger constant. Here, we do not assume such a minimization until \Cref{corr:stoq} and, consequently, drop the $S$ subscript {{to avoid cluttering the notation.}} 

Recall that we denote the spectrum of $H(s)+\Cut(s)$ as $\set{\lambda_j(s)},$ with $\lambda(s):=\lambda_0(s)<\lambda_1(s)\leq\lambda_2(s)\leq\cdots\leq\lambda_{N-1}(s).$ 
Further, denote the spectrum of $H_S(s)|_S$ by 
$\set{\mu_j(s)}$ with 
$\mu(s):=\mu_0(s)<\mu_1(s)\leq\mu_2(s)\leq\cdots$ and the spectrum of $H_{\overline{S}}(s)\vert_{\overline{S}}$ by $\set{\overline{\mu}_j(s)}$ with $\overline{\mu}(s):=\overline{\mu}_0(s)<\overline{\mu}_1(s)\leq\overline{\mu}_2(s)\leq\cdots$. 
{{We have implicitly assumed that the ground states of $H_S(s)$ and $H_{\overline{S}}(s)$ are non-degenerate. Let the eigenstates associated with these eigenvalues be denoted with vectors. For instance, the eigenstate associated with $\lambda$ is $\vec\lambda$.}}

We then define the spectral gaps, or simply gaps, of $(H(s)+\Cut(s))$, $H_S(s)$, and $H_{\overline{S}}(s)$ as
\begin{align}
    \gamma(s)&:=\lambda_1(s)-\lambda(s), \nonumber\\
    \Gamma_S(s)&:=\mu_1(s)-\mu(s), \nonumber\\
    \Gamma_{\overline{S}}(s)&:=\overline{\mu}_1(s)-\overline{\mu}(s),
\end{align}
respectively. From here on, whenever it is not an impediment to clarity, we drop the explicit $s$ dependence of operators to avoid clutter. The only quantities not considered, {{in general,}} to be $s$-dependent will be the subspaces $S$, $\overline{S}$. Given these definitions, we have the following proposition {{relating the various spectra defined above.}} Recall that we assumed that the Hamiltonian $H+\Delta$ is such that $\vec\lambda$ has support on $S$ and, consequently, the quantity $\lambda_S$ around which this proposition centers is well-defined. 
\begin{prop}\label{prop:evalfacts} 
 Let
\[
\lambda_S(s) := \frac{\inprod{\vec\gse}{\HS\vec{\gse}}}{\inprod{\vec\gse}{\projS\vec{\gse}}}.
\] 
Then, the following inequalities hold:
\begin{enumerate}
    \item \label{prop:evalfacts1} $\gse \leq \min\set{\gseS , \gsevalSbar}$,
    \item \label{prop:evalfacts2}$\gseS \leq \gsrS$, and 
    \item \label{prop:evalfacts3}$\gsrS \leq \gse + \h$.
\end{enumerate}
\end{prop}

\begin{proof}
    Without loss of generality, suppose that $\gseS = \min\{\gseS,\gsevalSbar\}$. \Cref{prop:evalfacts1} follows immediately from the fact that $\gseS$ is the smallest eigenvalue of $\HS \vert_S$ and Cauchy's interlacing theorem \cite[Corollary III.1.5]{Bhatia}. \Cref{prop:evalfacts2} follows from the following simple variational argument: note that
    \[
        \gse = \inf_{\vec\psi \in \mathcal{H}} \frac{\inprod{\vec\psi}{(H+\Delta)\vec\psi}}{\inprod{\vec\psi}{\vec\psi}}
        \leq \inf_{\vec\psi \in S} \frac{\inprod{\vec\psi}{(H+\Delta)\vec\psi}}{\inprod{\vec\psi}{\vec\psi}} = \gseS
        \leq \frac{\inprod{\projS\vec\gse}{(H+\Delta)\projS\vec\gse}}{\inprod{\projS\vec\gse}{\projS\vec\gse}}
        =\lambda_S.
   \]
 (This is also an alternative proof of \Cref{prop:evalfacts1}.) For \Cref{prop:evalfacts3}, the fact that $(H+\Delta) \vec\gse = \gse \vec\gse$, implies that $\gse\projS\vec{\gse}=\projS (H+\Delta)\vec{\gse}$. Thus, inserting the resolution of identity $\projS+\projSbar$, we have that
\begin{align*}
    \gse \inprod{\vec\gse}{\projS\vec{\gse}} &= \inprod{\vec{\gse}}{\projS (H+\Delta)\projS\vec{\gse}} + \inprod{\vec\gse}{\projS (H+\Delta)\projSbar\vec{\gse}}
     = \inprod{\vec{\gse}}{H_S\vec{\gse}} + \inprod{\vec\gse}{\projS \Delta\projSbar\vec{\gse}}, 
\end{align*}
which implies that
\begin{align*}
    \gse &= \gsrS + \frac{\inprod{\vec\gse}{\projS \Delta \projSbar\vec{\gse}}}{\inprod{\vec\gse}{\projS\vec{\gse}}}
    = \gsrS - \h\min\set*{\frac{\inprod{\vec\gse }{\projSbar \vec\gse}}{\inprod{\vec\gse}{\projS \vec\gse}},1},
\end{align*}
establishing \cref{prop:evalfacts3}. 
\end{proof}
It will also be convenient to define a a few additional operators. In particular, we define orthogonal projectors onto the ground spaces of $S, \overline{S}$:
\begin{align}\label{eq:pmu}
    P_\mu(s&:=P_{\mathrm{span}\set{\vec x|H_S\vec x=\mu \vec x}} &&
    P_{\overline{\mu}}(s):=P_{\mathrm{span}\set{\vec x|H_{\overline{S}}\vec x=\overline{\mu}\vec x}}.
\end{align}
 Given these, we also define
\begin{align}
    \Delta^\perp(s):={{\mathcal{M}^\perp}}\Delta{{\mathcal{M}^\perp}},
\end{align}
{{where $\mathcal{M}^\perp:=I-\mathcal{M}$ for $\mathcal{M}:=P_\mu+P_{\overline{\mu}}$. Thus, $\Delta^\perp(s)$ is the block-antidiagonal pertubation $\Delta$ projected away from the local ground spaces $P_\mu$ and $P_{\overline{\mu}}$.}}

Finally, we define a collection of time evolution operators generated by the various Hamiltonians we have defined. In particular, let $U(s)$, $U_\perp(s)$, and $V(s)$ be one-parameter families of unitary operators defined by initial value condition $U(0)=U_\perp(0)=V(0)=I$ and
\begin{align}\label{eq:unitaries}
    i\dot{U}(s)&=T(H(s)+\Cut(s))U(s), \nonumber \\
    i\dot{U}_\perp(s)&=T(H(s)+\Cut^\perp(s))U_\perp(s), \nonumber \\
    i\dot{V}(s)&=TH(s)V(s),
\end{align}
respectively. Finally, let the adiabatic unitary corresponding to $V(s),$ denoted $V_{\mathrm{ad}}(s),$ be the solution to $V_{\mathrm{ad}}(0)=I,$
\begin{align}\label{eq:Vad}
    i\dot{V}_{\mathrm{ad}}(s)=\left(TH(s)+i[\dot{P}_\mu(s), P_\mu(s)]\right)V_{\mathrm{ad}}(s),
\end{align}
where $\dot P_\mu, \dot P_{\overline{\mu}}$ denotes a derivative with respect to $s$.

\section{Motivating Scenario: The Time-Independent Case}\label{s:timeindependentcase}
Before diving into the general, time-dependent scenario outlined in ``Theorem''~\ref{sthm:mainresult}, it is helpful to consider the mathematically much simpler scenario of a static, time-independent Hamiltonian with $\dot H=\dot\Cut=0.$ By allowing us to brush the mathematical complications of time-dependence under the rug, this example lets us to better highlight the aims of our paper.

In this time-independent scenario adiabatic evolution by $H$ on $\vec\mu$ simply gives $e^{-i\mu t}\vec\mu.$ Thus, we are able to prove the following simple result bounding the difference between time evolution under $H$ and time evolution under the perturbed Hamiltonian $H+\Delta$:
\begin{theorem}\label{thm:timeindep}
    Consider a Hamiltonian $H+\Cut$ with $\dot H =\dot\Cut= 0.$ Then, for $\vec\mu,$ $h$ defined as above, 
    \[\left\Vert (e^{-i(H+\Cut)T}-e^{-i\mu T})P_\mu\right\Vert  \leq  2\sqrt{hT}.\]
\end{theorem}
\begin{proof}
    To start, consider decomposing $\vec\mu$ in the eigenbasis of $H+\Cut$ as $\vec\mu = \sum_i \alpha_i \vec\lambda_i$. It then holds that
    \begin{equation}\label{eq:decomp}
        \sum_i\vert \alpha_i\vert^2 \lambda_i = \mu.
    \end{equation}
    Thus, we have that 
     \begin{align*}
         e^{-i(H+\Cut)t}\vec\mu &= \sum_i \alpha_i e^{-i\lambda_i t}\vec\lambda_i= e^{-i\mu t}\sum_i \alpha_i e^{-i(\lambda_i-\mu) t}\vec\lambda_i.
    \end{align*}
    Using \cref{eq:decomp} allows us to write
    \begin{align*}
         \left(e^{-i((H+\Cut)-\mu)t}-I\right)\vec\mu &= \sum_i \alpha_i(e^{-i(\lambda_i-\mu)t}-1)\vec\lambda_i.
    \end{align*}
    Placing a norm on the above expression we obtain the following bound:
    \begin{align*}
         \left\Vert  \left(e^{-i((H+\Cut)-\mu)t}-I\right)\vec\mu\right\Vert^2 = \sum_i\vert \alpha_i\vert^2 \vert e^{-i(\lambda_i-\mu)t}-1\vert^2 &= 4\sum_{i}\vert \alpha_i\vert^2 \sin^2 \left(\frac{\lambda_i-\mu}{2} t\right) \\
         &\leq 2\sum_{i}\vert{\alpha_i}\vert^2 \vert \lambda_i-\mu\vert t \leq 4ht,
     \end{align*}
     where to get to the last line we used that $\lambda\leq \mu \leq \lambda+h$ (\cref{prop:evalfacts}) and used \cref{eq:decomp} to obtain that $\sum_{\lambda_i\geq \mu} \vert \alpha_i\vert^2 (\lambda_i-\mu)= -\sum_{\lambda_i<\mu}\vert\alpha_i\vert^2(\lambda_i-\mu)$; consequently, it holds that
     \begin{align*}
     \sum_{i}\vert \alpha_i\vert^2 \vert \lambda_i-\mu\vert &= \sum_{\lambda_i>\mu} \vert \alpha_i\vert^2 (\lambda_i-\mu)-\sum_{\lambda_i<\mu} \vert \alpha_i\vert^2 (\lambda_i-\mu)\\
     &= -2\sum_{\lambda_i<\mu} \vert \alpha_i\vert^2(\lambda_i-\mu)\leq 2h\sum_{\lambda_i<\mu} \vert \alpha_i\vert^2\leq 2h. 
     \end{align*}
     The result follows immediately.
\end{proof}
Crucially, the result in \cref{thm:timeindep} does not depend explicitly on $\norm{\Delta}$---that is, while $h$ depends on $\Delta$, generically, $h\ll \norm{\Delta}$. Thus, such bounds capture the fact that there are scenarios, especially in computational problems of the sort analyzed in adiabatic quantum computing, where even though $\norm{\Delta}$ is large, the errors introduced by ignoring its contribution can be comparably small. 

Furthermore, the quantity $h$ is closely related to the spectral gap of the full Hamiltonian and is quadratically close to it for most systems which can presently be engineered\footnote{so-called stoquastic Hamiltonians, an unfortunate name for non-positive Hamiltonians}. (See \cref{eq:stoq_bnd}.) Computationally difficult problems typically come equipped with a spectral gap that is polynomially-to-exponentially small in the Hilbert space dimension $N$. In many of these cases, an appropriately chosen of $S$ yields $h \in \mathcal{O}(N^{-1})$ despite $\norm{\Delta(s)} \in \Omega(1)$. Our bounds are most applicable and useful in these scenarios.

\section{Main Result}\label{s:results}
 
With all necessary definitions and some intuition from the time independent case in place, we are ready to present our main result. 
We begin with some assumptions related to {{the spectral structure of $H+\Cut$, $H_S$, and $H_{\overline{S}}$ and a description of the intuition behind them:}}
\begin{assmpn}\label{assmpn:tbd}
   Let $\min\{\mu_1,\overline\mu_1\}> \max\{\mu,\overline\mu\}.$ 
\end{assmpn}
That is, we assume the ground state energies of the Hamiltonians $H_S$ and $H_{\overline{S}}$ are ``close'' relative to the local gaps $\Gamma_S$ and $\Gamma_{\overline{S}}$. This assumption allows us to introduce the following perturbed Hamiltonian while incurring a minimal cost to our error bound:
\begin{equation}\label{eq:H'}
        H':=H+(\mu-\overline{\mu})P_{\overline{S}},
    \end{equation}
which has $\mu=\overline{\mu}$. Such a property will end up being useful for technical reasons in \cref{lemma:perturbadiab}. Note that $H'$ is still block diagonal, and, from \cref{prop:evalfacts} (Items~\ref{prop:evalfacts1} and \ref{prop:evalfacts3}), we have
\begin{align}\label{eq:size-of-H-prime-perturb}
    \mu-\overline{\mu}\leq (\lambda+h)-\lambda = h,
\end{align}
so this perturbation is only of size $\mathcal{O}(h)$. Consequently, 
we will see that replacing $H$ with $H'$ introduces a relatively inconsequential additive factor to our error bounds. 
Importantly, this term has no impact on the ``local'' eigenstates $\{\vec\mu_j\}_j$ or $\{\vec{\overline{\mu}}_j\}_j$. In fact, all quantities local to $S$ are invariant under this perturbation (e.g. $P_\mu$ and $H_S$), as well as many quantities local to $\overline{S}$ (e.g. $P_{\overline{\mu}}$). Consequently, when acting on $P_\mu$, the perturbed unitary defined by the initial value condition $V'_{\mathrm{ad}}(0)=I$ and
\begin{align}\label{eq:unitaries_perturbed}
    i\dot{V}_{\mathrm{ad}}'(s)&=\left(TH'(s)+i[\dot{P}_\mu(s), P_\mu(s)]\right)V'_{\mathrm{ad}}(s),
\end{align}
acts identically to its unperturbed counterpart defined in \cref{eq:Vad}. On the other hand, some operators that act on the space $\overline{S}$ are changed by this perturbation and the corresponding changes in their norms will have to be accounted for in the analysis that follows. 
While this perturbation eases the analysis, we do suspect it to be a technical artifact that could potentially be removed.

\begin{assmpn}\label{assmpn:tbd2}
    Let $\Vert\Cut\Vert \leq c \min\{\Gamma_S,\Gamma_{\overline S}\}$ where $c\leq 1-\frac{h}{\min\{\Gamma_S,\Gamma_{\overline S}\}}\leq 1$.
\end{assmpn}
This assumption amounts to the saying that the size of the off-block diagonal perturbation $\Delta$ is sufficiently small. However, we only require that $\norm{\Delta}$ be of the same asymptotic order as the local gaps $\Gamma_S,\Gamma_{\overline{S}}$, which are generically much larger than $h$ (the ``small'' parameter in our problem). 
To make this statement explicit, it is also helpful to add two additional assumptions, which, while not necessary for our proofs, specify the set of problems where our bounds are useful for analyzing asymptotics:
\begin{assmpn}\label{assmpn:tbd3}
    Let $\frac{\kappa}{(1-c)\min\{\Gamma_S,\Gamma_{\overline{S}}\}}=\Theta\left(h^{2(\alpha-1)}\right)$ for some $\alpha>0$, where $\kappa:=\Vert H+\Delta-\lambda\Vert=\lambda_{N-1}-\lambda$.
\end{assmpn}
\cref{assmpn:tbd3} specifies how the local gaps compare to the spectral properties of $H+\Delta$. To get intuition, note that it is satisfied if $\kappa=\Theta(\min\{\Gamma_S,\Gamma_{\overline{S}}\}$.
\begin{assmpn}\label{assmpn:tbd4}
    Let $P_\downarrow$ be the spectral projection of $H+\Delta$ onto $(-\infty, \lambda+2h],$ $P_\uparrow:=I-P_\downarrow,$ and $\epsilon_T := \max\set*{\norm*{P_\uparrow U(s)P_\downarrow}, \norm*{\mathcal{M}^\perp U_\perp'(s)\vec\mu} }$. Then let $\epsilon_T\sqrt{\kappa}=\Theta(h^{\beta-\frac{1}{2}})$ for some $\beta>0$. 
\end{assmpn}
\cref{assmpn:tbd4} amounts to assuming that there is a sizeable gap $\Gamma_\updownarrow$ between the low energy spectrum with spectral projection $P_\downarrow$ of $H+\Delta$ and the high energy spectrum with spectral projector $P_\uparrow$. To see this, we note that $\norm{P_\uparrow U(s) P_\downarrow}$ is precisely what is upper bounded by the usual adiabatic theorem applied to the subspace $P_\downarrow$. Roughly, by the ``folklore'' adiabatic theorem, this means $\norm{P_\uparrow U(s) P_\downarrow}=\mathcal{O}\big(1/(T\Gamma_\updownarrow^2)\big)$. For a fully rigorous statement we can use e.g. Ref.~\cite{jansen2007bounds} (see, also, \cref{thm:JRS}). The second term in the definition of $\epsilon_T$ is upper bounded precisely by a local adiabatic theorem of the sort we prove in this paper. Roughly, $\norm*{\mathcal{M}^\perp U_\perp'(s)\vec\mu}=\mathcal{O}\big(1/(T\min\{\Gamma_S^2,\Gamma_{\overline{S}}^2\})\big)$. See \cref{cor:local-adiabtic-thm} for the rigorous statement. Or, rearranging, these expressions, this assumption says roughly that the minimum of the local gaps $\Gamma_S, \Gamma_{\overline{S}}$ and the gap $\Gamma_\updownarrow$ is $\Theta\Big(\sqrt{h^{1/2-\beta}\sqrt{\kappa}T^{-1}}\Big)$.
The implications of all four assumptions for the spectra of $H+\Delta$, $H_S$, and $H_{\overline{S}}$ are summarized by the diagram in \cref{fig:spectrum}.

\begin{figure}
\centering\includegraphics[width=0.8\textwidth]{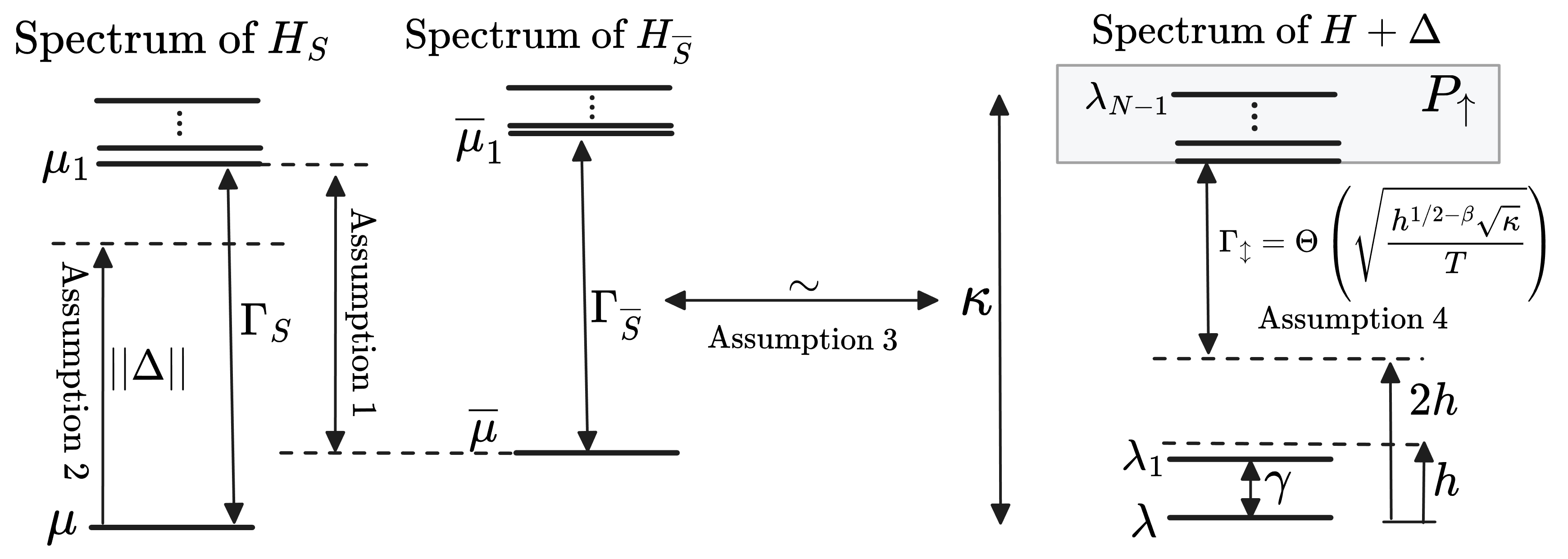}
\caption{A diagram indicating what \cref{assmpn:tbd,assmpn:tbd2,assmpn:tbd3,assmpn:tbd4} imply about the spectra of $H_S$, $H_{\overline{S}}$, and $H+\Delta$. }\label{fig:spectrum}
\end{figure}

While we do not believe that \cref{assmpn:tbd,assmpn:tbd2} to be overly restrictive, we remark that these assumptions may not be strictly necessary. In fact, we suspect that they may be able to be relaxed, while proving similar results. However, we were unable to find a rigorous way to remove them, so we leave this as a problem for future work. 
These minor caveats out of the way, our main result is as follows:
\begin{theorem}[Main result]\label{thm:mainresult}
    Let $H(s)$ be a one parameter, twice-differentiable family of Hamiltonians respecting the grading $\mathcal{H}=S\oplus\overline{S}$ and $\Delta(s)$ be a block-antidiagonal perturbation, where $s:=t/T\in[0,1]$. When acting on the ground space $P_\mu$ of $S$, the difference between the unitary evolution $U(1)$ generated by the Hamiltonian $H(s)+\Cut(s)$ acting for time $T$ and the adiabatic unitary evolution $V_{\mathrm{ad}}(1)$ generated by the block-diagonal Hamiltonian $H(s)$ is bounded as
    \[
    \norm{(U(1)-V_{\mathrm{ad}}(1))P_{\mu}(0)}\leq \max_{s\in[0,1]} \left[B\sqrt{hT}+\frac{C}{\eta T}\right],
    \]
    where    
\begin{align*}
B=2\sqrt{1 +  2\sqrt\frac{\kappa}{(1-c)\min\{\Gamma_S,\Gamma_{\overline S}\}} + \sqrt{\frac{\kappa}{h}}\epsilon_T}
\end{align*}
\begin{align*}
    C&=    \frac{2\Vert\dot{H}_{S}\Vert}{\Gamma_S^2}+\frac{\Vert\ddot{H}_{S}\Vert}{\Gamma_S^2}+\frac{10\Vert\dot{H}_{S}^2\Vert}{\Gamma_S^3}+\frac{\Vert{\dot{H}_{S}}\Vert}{\eta\Gamma_S^2}\left(\frac{\Vert{\dot{\Cut}^\perp\Vert}}{\min\{\Gamma_S,\Gamma_{\overline{S}}\}}+4\Vert\Cut^\perp\Vert \max\left\{\frac{\Vert{\dot{H}_{S}}\Vert}{\Gamma_S^2},\frac{{2\Vert\dot{H}_{\overline{S}}\Vert+\Vert \dot{H}_{S}\Vert}}{\Gamma_{\overline{S}}^2}\right\}\right).
\end{align*}
and $\epsilon_T := \max\{\norm*{P_\uparrow U(s)P_\downarrow}, \norm*{\mathcal{M}^\perp U_\perp'(s)\vec\mu} \},$ $   \eta:=1-\Vert{\Delta^\perp}\Vert/\min\{\Gamma_S,\Gamma_{\overline{S}}\},$, $\kappa:=\lambda_{N-1}-\lambda$, and $c:=\Vert \Cut\Vert /\min\{\Gamma_S,\Gamma_{\overline S}\}$ as defined in \cref{assmpn:tbd2,assmpn:tbd3,assmpn:tbd4}. The $\max_{s\in[0,1]}$ is understood to apply to each term individually.
\end{theorem}
\begin{remark} ``Theorem'' \ref{sthm:mainresult} holds as a first order reduction of \cref{thm:mainresult} {{when \cref{assmpn:tbd3,assmpn:tbd4} hold and, thus, $B=\mathcal{O}(1)$ and $C$ reduces to something quite resembling the usual adiabatic error applied to the subspace $S$. This is the limit where our results can be usefully applied i.e. to analyze quantum annealing.}}
\end{remark}
\begin{proof}
    The proof proceeds in two steps. 
    First, we define a unitary $U_\perp'(1)$ by the initial value condition $U_\perp'(0)=I$ and
    \begin{equation}\label{eq:Uperpprime}
    i\dot{U}'_\perp(s)=T(H'(s)+\Cut^\perp(s))U'_\perp(s),
    \end{equation}
    where $H'$ is as defined in \cref{eq:H'}.
    We then bound the difference between unitary evolution under $U'(1)$ and $U_\perp'(1)$ to yield
    \[
    \Vert (U(1)-U_{\perp}'(1))P_{\mu}(0)\Vert\leq B\sqrt{hT}.
    \]
    This is shown in \cref{lemma:cheegerperturb}, derived in \cref{s:cheegerperturb}, below. We then show the following bound
    \[
    \Vert (U_\perp'(1)-V_{\mathrm{ad}}(1))P_{\mu}(0)\Vert\leq \frac{C}{\eta T},
    \]
    in \cref{lemma:perturbadiab} in \cref{s:perturbadiabatic}. Applying the triangle inequality, the result follows.
\end{proof}

\section{Technical Details}\label{s:details}
In this section, we provide the proofs of the lemmas that lead to our main result in \cref{thm:mainresult}. We begin in \cref{s:spectral_bounds} by deriving a collection of results bounding the spectral norms of many of the quantities defined in \cref{s:mathematicalprelim}. In \cref{s:cheegerperturb}, we derive bounds on the leakage error and, in \cref{s:perturbadiabatic}, we derive the local version of the (perturbed) adiabatic theorem. 

\subsection{Spectral Bounds}\label{s:spectral_bounds}
The results in this section are a collection of bounds relating to the various spectra and spectral projections defined above. 
Some may be of independent interest. 

\begin{lemma}\label{prop:Pmunorms}
    Let $P(s)$ be the spectral projection of a $C^2([0,1])$ one-parameter family of Hamiltonians $H(s)$ with spectrum $\sigma(H(s))$ onto some $\alpha(s)\subseteq \mathbb{C}$. If $\sigma(H(s))\backslash\alpha(s) \neq \emptyset$ let $\mathrm{dist}(\alpha(s), \sigma(H(s))\backslash \alpha(s))\geq \delta(s)>0$, where $\delta(s)\in\mathbb{R}$.
    Then, letting $m(s)$ be the number of eigenvalues (counting multiplicity) contained in $\alpha(s),$
    \begin{align*}
     \Vert{\dot P}\Vert &\leq \frac{\Vert{ \dot{H}}\Vert}{\delta}, \\
    \Vert{(I-P)\ddot P P}\Vert&\leq \sqrt{m(s)}\frac{\Vert{\ddot H}\Vert}{\delta} + 4m(s)\frac{\Vert{\dot H}\Vert^2}{\delta^2}.\end{align*}
    If  $\sigma(H(s))\backslash\alpha(s) = \emptyset$, then $\norm{P}=\Vert{\dot{P}}\Vert=\Vert{\ddot{P}}\Vert=0.$
\end{lemma}
{
\begin{remark}
    In some cases, these bounds can be significantly tightened. However, this will usually require more specific knowledge about the Hamiltonian family than we would like to assume. As we are interested in applying our results to computational problems, we would like to be able to use bounds that do not require detailed knowledge of problem solutions a priori.
\end{remark}
}
\begin{proof}
    The case where $\sigma(H(s))\backslash\alpha(s) = \emptyset$ is trivially true. For the other case, note that
    \[\dot{P}(s) := \lim_{d\to 0}\frac{P(s+d)-P(s)}{d}.\] By the Davis-Kahan theorem \cite[Thm. VII.3.1]{Bhatia},\footnote{{{The Davis-Kahan $\sin\theta$ theorem says that for Hermitian operators $A,B$ and for two intervals $S_1=[a,b]\subset\mathbb{R}$ and  $S_2$ the complement in $\mathbb{R}$ of the interval $[a-\delta,b+\delta]$ then for all unitarily invariant norms
    \[
    \norm{EF}\leq \frac{1}{\delta} \norm{E(A-B)F}\leq \frac{1}{\delta}\norm{A-B},
    \]
    where $E$ ($F$) is the orthogonal projection of $S_1$ ($S_2$) onto the subspace spanned by the eigenvectors of $A$ ($B$). The inequalities below are a direct application of this theorem with $A=H(s+d)$, $B=H(s)$, $E=P^\perp(s+d)$, $F=P(s)$, and $\delta=\mathrm{dist}(\alpha(s),\sigma(H(s))\setminus\alpha(s))$.}}} we have that 
    \begin{align*}\Vert P(s+d)-P(s)\Vert &= \Vert P^\perp(s+d)P(s)\Vert \leq \frac{\Vert P^\perp(s+d)(H(s+d)-H(s))P(s)\Vert}{\mathrm{dist}(\alpha(s+d), \sigma(H(s))\backslash \alpha(s))}\\
    &\leq \frac{\Vert H(s+d)-H(s))\Vert}{\mathrm{dist}(\alpha(s+d), \sigma(H(s))\backslash \alpha(s))}\\
    \end{align*}
    Taking the limit, we have that 
    \begin{align*}
        \Vert \dot{P}(s)\Vert &= \lim_{d\to 0} \frac{\Vert P(s+d)-P(s)\Vert}{d}\\
        &\leq \lim_{d\to 0}\frac{1}{d}\frac{\Vert H(s+d)-H(s)\Vert}{\mathrm{dist}(\alpha(s+d), \sigma(H(s))\backslash \alpha(s))}\\
        &= \left(\lim_{d\to 0}\frac{\Vert H(s+d)-H(s)\Vert} {d}\right)\left(\lim_{d\to 0} \frac{1}{\mathrm{dist}(\alpha(s+d), \sigma(H(s))\backslash \alpha(s))}\right)\\
        &= \frac{1}{\delta(s)}\left\Vert \lim_{d\to 0}\frac{H(s+d)-H(s)}{d}\right\Vert= \frac{\Vert\dot H(s)\Vert}{\delta(s)},
    \end{align*}
    where we have that both limits on the third line exist because $H\in C^2([0,1]),$ implying that it is continuously differentiable and its eigenvalues are continuous with respect to $s$. The bound on $(I-P)\ddot P P$ is proven in \cite[Lemma 8]{jansen2007bounds}.
\end{proof}

Note that the bound on $\norm{\dot P(s)}$ is a factor of $\sqrt{m}$ tighter than that considered in \cite[Lemma 8]{jansen2007bounds}. We also believe the bound on $\norm{\ddot{P}}$ can be similarly improved to at least remove the factor of $m$ multiplying $\norm{\dot{H}}$, but we leave this for future analysis.
We will be applying this lemma to $P_\mu$ and $P_{\overline\mu}$, in the form of the following corollary:
\begin{cor}
    Let $P_\mu, P_{\overline\mu}$ $H_S,$ $H_{\overline{S}},$ $\Gamma_S,$ and $\Gamma_{\overline{S}}$ be defined as above. Then
    \[\Vert \dot P_\mu\Vert \leq \frac{\Vert \dot H_S\Vert}{\Gamma_S}\]
    and 
    \[\Vert (P_S-P_\mu)\ddot{P}_\mu P_\mu\Vert \leq \frac{\norm{\ddot H_S}}{\Gamma_S^2}+4\frac{\norm{\dot H_S}}{\Gamma_S}.\] The same bounds hold if $\mu$ and $S$ are replaced by $\overline\mu$ and $\overline S$ respectively.
\end{cor}
Using the projectors $P_\mu$ and $P_{\overline{\mu}}$ we are now ready to proceed with \cref{lem:expected_norm}. This result bounds the transitions between the ground spaces $P_{\mu}$ and $P_{\overline{\mu}}$ of the projected Hamiltonians $H_S$ and $H_{\overline{S}}$ using the Cheeger ratio. We will use this to quantitatively understand the leakage between these subspaces. In particular, our results will be most useful when the  leakage between subspaces is sufficiently small.

\begin{lemma}\label{lem:expected_norm}
    Let $P_\mu$ and $P_{\overline{\mu}}$ be the projectors onto the eigenvectors of $H$ corresponding to eigenvalues $\mu$ and $\overline\mu$ respectively. If $\mathcal{M} = P_\mu + P_{\overline\mu}$, then
    \[
        \norm{\mathcal{M}\Delta \mathcal{M}} \leq  h,
    \]
    where $h$ is the Cheeger ratio (\cref{def:cheeger}).
\end{lemma}
\begin{proof}
    $\mathcal{M}\Cut \mathcal{M}=P_\mu \Cut P_{\overline{\mu}} + P_{\overline\mu} \Cut P_{{\mu}}$ is block-antidiagonal, so $\Vert \mathcal{M}\Cut \mathcal{M}\Vert = \Vert P_\mu \Cut P_{\overline{\mu}}\Vert .$ Thus,  
    \begin{align*}
        \Vert \mathcal{M}\Cut\mathcal{M}\Vert &= \Vert P_\mu \Cut P_{\overline{\mu}}\Vert= \Vert P_\mu (H+\Delta-\lambda)P_{\overline\mu}\Vert\\
       &\leq  \sqrt{ \Vert P_\mu (H+\Delta-\lambda)P_\mu \Vert \Vert P_{\overline\mu} (H+\Delta-\lambda)P_{\overline\mu} }\Vert\\
       &= \sqrt{(\mu-\lambda)(\overline\mu-\lambda)},
    \end{align*} where to get to the second line we used Cauchy-Schwarz and that $H+\Cut-\lambda$ is positive-semidefinite. By \cref{prop:evalfacts}.\ref{prop:evalfacts2}-\ref{prop:evalfacts3}, we thus establish our result.
\end{proof}

We will also require the following looser bound. 
\begin{lemma}\label{lemma:deltaonPmu}
Let $H$ and $\Cut$ be as given above with $\lambda$ the smallest eigenvalue of $H+\Cut$ and $\mu$ the smallest local eigenvalue of $H_S$  with $P_\mu$ the corresponding projector. Then it holds that {
\[
\Vert \Delta P_\mu\Vert \leq \sqrt{h\kappa},
\]
}
where $h$ is the Cheeger ratio (\cref{def:cheeger}) and $\kappa:= \norm{H+\Delta-\lambda}=\lambda_{N-1}-\lambda$.
\end{lemma}
\begin{proof}
    First, {
    \begin{align}
        \norm{\cut \gsprojS} &= \norm{\projSbar({H}+\Cut-\mu)\gsprojS} = \norm{\projSbar({H}+\Cut-\lambda)\gsprojS}.\label{eq:aa}
    \end{align}}
    By \Cref{prop:evalfacts}.\ref{prop:evalfacts1}, ${H}+\Cut - \gse$ is positive semidefinite. Thus, $\sqrt{H+\Cut - 
 \gse}$ is well-defined. Hence,
    \begin{align}
        \norm*{({H}+\Cut-\gse)\gsprojS}&\leq \norm*{{\projSbar}\sqrt{{H}+\Cut-\gse}}\norm*{\sqrt{{H}+\Cut-\gse}\gsprojS} \nonumber
        \\
        &=\norm*{{\projSbar}\sqrt{{H}+\Cut-\gse}}\sqrt{\gseS-\gse} \nonumber
        \\&\leq \norm*{{\projSbar}\sqrt{{H}+\Cut-\gse}}\sqrt{\h}, \label{eq:ab}
    \end{align}
    where the final inequality follows from \Cref{prop:evalfacts}.\ref{prop:evalfacts3}. Combining \cref{eq:aa} with \cref{eq:ab}, 
    
    and applying the definition of $\kappa$ we have that $\norm{\cut \gsprojS} \leq \sqrt{\h\kappa}.$
\end{proof}

\begin{cor}\label{cor:deltaonPmubar}
It also holds that 
$\Vert \Delta P_{\overline{\mu}}\Vert \leq \sqrt{h\kappa},$ and $\Vert \Delta \mathcal{M} \Vert \leq \sqrt{h\kappa}. $

\end{cor}
\begin{proof}
    The proof of the first statement is identical to \cref{lemma:deltaonPmu} under interchange of the subspaces $S$, $\overline{S}$ as the proof is agnostic to which subspace is called $S$ and which is called $\overline{S}$. The same holds for \cref{prop:evalfacts}. The second statement immediately follows from the fact that
    \[
    \Vert \Delta \mathcal{M}\Vert \leq \max\left\{ \Vert \Delta P_\mu\Vert, \Vert \Delta P_{\overline{\mu}} \Vert\right\}.
    \]
\end{proof}

\begin{lemma}\label{lemma:projection_approximation}
Let $\lambda$ be the lowest eigenvalue of $H + \Cut$, as defined above, and let $P_\lambda$ be the orthogonal projector onto its ground state $\vec\lambda$. Let $(\gseS, \gsprojS)$ be the lowest eigenvalue-orthogonal projector pair of $\HS$, as defined above. Assuming that $P_S\gsprojS \neq 0$, then,
\[
\frac{\norm{\gsproj \gsprojSperp}}{\norm{\gsproj P_S}} \leq \sqrt{\frac{\gsrS-\gseS}{\gapS}},
\]
where $\gsprojSperp := \projS-\gsprojS$.
\end{lemma}

\begin{proof}
Using that $\projS$ is an orthogonal projector and defining $P_{\mu_j}$ as the orthogonal projector onto the eigenspace of $\HS$ with eigenvalue $\mu_j$ it holds that
\begin{align}
    \norm{\gsproj P_S}^2\left(\gsrS-\gseS\right) &= \norm{\gsproj\projS\gsproj}\left(\gsrS-\gseS\right)=\norm{\gsproj\projS({H}-\gseS)\projS\gsproj} =\norm{\sqrt{\HS-\gseS}\projS\gsproj}^2\nonumber
    \\&=\norm{\sqrt{\HS-\gseS}(P_\mu + P_\mu^\perp)\gsproj}^2\nonumber
    \geq\norm{\sqrt{\HS-\gseS}P_\mu^\perp\gsproj}^2\nonumber
    \geq \gapS \norm{\gsprojSperp\gsproj}^2 \nonumber
\end{align}
where the first inequality follows from the fact that both terms are positive semidefinite and for positive semidefinite operators $A,B$ it holds that $\norm{A+B}\geq\max\{\norm{A},\norm{B}\}$. Thus,
\begin{equation}
 \frac{\gsrS-\gseS}{\gapS} \geq \frac{\norm{\gsproj\gsprojSperp}^2}{\norm{\gsproj \projS}^2}.
\end{equation}
\end{proof}

\begin{cor}
With all terms defined as in \cref{lemma:projection_approximation}, 
\[
\frac{\norm{\gsproj \gsprojSperp}}{\norm{\gsproj P_S}} \leq \sqrt{\frac{h}{\gapS}}.
\]
Further, let $\mathcal{M}^\perp:=I-P_\mu-P_{\overline\mu}.$ Then
\[\Vert P_\lambda \mathcal{M}^\perp \Vert \leq \sqrt\frac{2h}{\min\{\Gamma_S,\Gamma_{\overline S}\}}.\]
\end{cor}
\begin{proof}
This first statement comes from using \cref{prop:evalfacts} to relax the inequality in \cref{lemma:projection_approximation}. Applying this statement and the definition of $\mathcal{M}^\perp$ gives the second statement.
\end{proof}

\begin{lemma}\label{cor:PionM}

Let $P_\downarrow$ be a spectral projection of $H+\Delta$ onto $(-\infty,\lambda+2h]$.

 $\mathcal{M}:= P_\mu+P_{\overline\mu},$ and let all other quantities be defined as above. Then, 
    \begin{align*}
        \Vert P_\uparrow \mathcal{M}\Vert &\leq {2}\sqrt\frac{h}{(1-c)\min\{\Gamma_S,\Gamma_{\overline S}\}},
    \end{align*}
    where $c:=\Vert \Cut\Vert /\min\{\Gamma_S,\Gamma_{\overline S}\}$ is as defined in \cref{assmpn:tbd2}.
\end{lemma}
\begin{proof}
    { First, note that
\begin{align*}
    \Vert \mathcal{M}P_\uparrow\Vert 
    &\leq \Vert P_\mu P_\uparrow\Vert+  \Vert P_{\overline\mu}P_\uparrow\Vert.
\end{align*}
Without loss of generality, consider $\Vert P_\mu P_\uparrow\Vert.$ We have that 
\[\Vert P_\mu P_\uparrow\Vert^2 = \mathrm{Tr}[P_\mu P_\uparrow P_\mu].\] We also have by \cref{prop:evalfacts} that 
\[    \mathrm{Tr}\left[(H+\Delta-\lambda)P_\mu \right]\leq h.
\]
However, we further have 
\begin{align*}
\mathrm{Tr}\left[(H+\Delta-\lambda)P_\mu \right] &=  \mathrm{Tr}\left[(H+\Delta-\lambda)P_\uparrow P_\mu \right]+\mathrm{Tr}\left[(H+\Delta-\lambda)P_\downarrow P_\mu \right]\\&\geq(\lambda_\uparrow-\lambda)\mathrm{Tr}[P_\mu P_\uparrow P_\mu],
\end{align*}
$\lambda_\uparrow := \min_j \{\lambda_j | \lambda_j > \lambda+2h\},$ as $(H+\Delta-\lambda)P_\uparrow\succcurlyeq (\lambda_\uparrow-\lambda)P_\uparrow,$ $(H+\Delta-\lambda)P_\downarrow\succcurlyeq0,$ and the expected value of a PSD matrix is nonnegative.
Thus, we have that 
\[\Vert P_\mu P_\uparrow\Vert \leq \sqrt\frac{h}{\lambda_\uparrow-\lambda}.\]
This also holds for $P_{\overline{\mu}}.$ Thus, we have that 
\[\Vert\mathcal{M}P_\uparrow\Vert \leq 2\sqrt\frac{h}{\lambda_\uparrow-\lambda}.\]
}

    Now, we move to bound $\lambda_\uparrow-\lambda.$ As $\lambda_1\leq \lambda+ 2h$, $\lambda_\uparrow-\lambda \geq  \lambda_2 - \lambda$. Then, by Weyl's inequality and \cref{assmpn:tbd}, we have that $\vert \lambda_2-\min\{\mu+\Gamma_S,\overline\mu+\Gamma_{\overline S}\}\vert \leq \Vert \Cut\Vert.$ Thus, $\lambda_2 \geq \min\{\mu+\Gamma_S,\overline\mu+\Gamma_{\overline S}\}-\Vert \Cut\Vert .$ Using \cref{assmpn:tbd2}, we thus have that $\lambda_\uparrow-\lambda \geq \lambda_2-\lambda \geq (1-c)\min\{\Gamma_S,\Gamma_{\overline{S}}\}$. Plugging this in, we arrive at our result.
\end{proof}

At this point we have assembled the basic tools with which we can proceed with proving the lemmas that form the main substance of \cref{thm:mainresult}.

\subsection{Bounding the Difference Between $U(1)$ and $U'_\perp(1)$} \label{s:cheegerperturb}
In this section, we prove the second component of our main result (\cref{thm:mainresult}). This contribution to the bound will scale with $\sqrt{T}$ and can intuitively be interpreted as the error introduced due to leakage of the ground state from $S$ to $\overline{S}$. This is why this error term grows with the evolution time $T$.
In particular, we prove the following lemma.
{ 
\begin{lemma}\label{lemma:cheegerperturb} 
Consider all terms as defined above. Then
\begin{align*}
    \Vert (U(1)-U_\perp'(1))P_\mu(0)\Vert ^2 &\leq \max_{s\in[0,1]} 4hT\left(1 +  2\sqrt\frac{\kappa}{(1-c)\min\{\Gamma_S,\Gamma_{\overline S}\}} + \sqrt{\frac{\kappa}{h}}\epsilon_T\right)
\end{align*}
where $\epsilon_T := \max\set*{\norm*{P_\uparrow U(s)P_\downarrow}, \norm*{\mathcal{M}^\perp U_\perp'(s)\vec\mu} }$, $P_\downarrow$, $\kappa:= \norm{H+\Delta-\lambda}$
, and $c:=\Vert \Cut\Vert /\min\{\Gamma_S,\Gamma_{\overline S}\}$ as defined in \cref{assmpn:tbd2,assmpn:tbd3,assmpn:tbd4}.
\end{lemma}
}

{\begin{remark}
    If \cref{assmpn:tbd2,assmpn:tbd3,assmpn:tbd4} hold then this bound has asymptotic scaling $\mathcal{O}(h^\nu T)$ where $\nu=\min\{\alpha,\beta,1\}\!>0$ for $\alpha,\beta$ as defined in \cref{assmpn:tbd3,assmpn:tbd4}.
\end{remark}}

\begin{proof}
First, using the fundamental theorem of calculus and the definitions of $U(s)$ and $U'_\perp(s)$ (\cref{eq:unitaries} and \cref{eq:Uperpprime}, respectively):
\begin{align*}
    U^\dagger(1)U'_\perp(1)-I &= \int^1_0 ds \frac{d}{ds}\left(U^\dagger(s)U'_\perp(s) \right)
    \\&=i T \int^1_0 ds U^\dagger(s)\left( H+\Delta -H - \Delta^\perp -(\mu-\overline{\mu})P_{\overline S}\right)U'_\perp(s)
    \\&=i T \int^1_0 ds U^\dagger(s)\left(\Delta - \Delta^\perp +(\overline\mu-\mu)P_{\overline S}\right)U'_\perp(s).
\end{align*}
Letting $\vec\mu=\vec\mu(0)$ for brevity, we also have that  
\begin{align}\label{eq:temp2}
\Vert (U(1)-U'_\perp(1))P_{\mu}(0)\Vert^2 &= \vec\mu^\dagger (U(1)-U'_\perp(1))^\dagger(U(1)-U'_\perp(1))\vec\mu \nonumber\\
&= 2-2\mathrm{Re}(\vec\mu^\dagger U^\dagger(1)U'_\perp(1)\vec\mu)\leq 2\vert \vec\mu^\dagger U^\dagger(1)U'_\perp(1)\vec\mu-1\vert.
\end{align}
Thus, we have that 
\begin{align}
    \abs{\vec\mu^\dagger U^\dagger(1)U_\perp'(1)\vec\mu-1} &= T \abs{ \int^1_0 ds \vec\mu^\dagger U^\dagger(s){ \left(\Delta - \Delta^\perp +(\overline\mu-\mu)P_{\overline S}\right)U'_\perp(s)\vec\mu }}
    \nonumber\\&=T \abs{ \int^1_0 ds \vec\mu^\dagger U^\dagger(s){ \left(\mathcal{M}^\perp\Delta \mathcal{M} + \mathcal{M}\Delta\mathcal{M}^\perp + \mathcal{M} \Delta \mathcal{M} +(\overline\mu-\mu)P_{\overline S}\right)U'_\perp(s)\vec\mu}}
    \nonumber\\&\leq T \abs{ \int^1_0 ds \vec\mu^\dagger U^\dagger(s){ \left(\mathcal{M}^\perp\Delta \mathcal{M} + \mathcal{M}\Delta\mathcal{M}^\perp + \mathcal{M} \Delta \mathcal{M}\right)}}+ hT
    \nonumber\\&\leq T \abs{ \int^1_0 ds \vec\mu^\dagger U^\dagger(s){ \left(\mathcal{M}^\perp\Delta \mathcal{M} + \mathcal{M}\Delta\mathcal{M}^\perp \right)}U'_\perp(s)\vec\mu} + 2hT, \label{eq:MdeltaMterm}
\end{align}
where the first inequality comes from recalling that $\vert \overline\mu-\mu\vert < h$ and the final line follows from applying \cref{lem:expected_norm} to the last term in the integrand.

Suppressing all dependencies upon $s$ for clarity we now bound the remaining terms. First,
\begin{align}
    \abs{\vec\mu^\dagger {U}^\dagger\mathcal{M}\Delta\mathcal{M}^\perp U_\perp'\vec\mu}
    &\leq \norm{\mathcal{M}\Delta}\norm{\mathcal{M}^\perp U'_\perp\vec\mu} \nonumber
    \leq {\epsilon_T\sqrt{h \kappa}},\label{eqn:second_term_bound}
\end{align}
where, in the final line, we have applied \cref{cor:deltaonPmubar} and the definition of $\epsilon_T$. The remaining term will take a bit more work. To start we have that
\begin{align}
    \abs{\vec\mu^\dagger {U}^\dagger\mathcal{M}^\perp\Delta \mathcal{M}U'_\perp\vec\mu} &\leq \norm{\mathcal{M}^\perp {U}\vec\mu}\norm{\Delta\mathcal{M}} \nonumber
    \leq \norm{\mathcal{M}^\perp {U}\vec\mu}{\sqrt{h\kappa}},
\end{align}
where we have, again, applied \cref{cor:deltaonPmubar}. We are now left with bounding $\Vert{\mathcal{M}^\perp {U}\vec\mu}\Vert$. To this end, observe that
\begin{equation}\label{eq:lem5a}  
    \norm{P_\uparrow \mathcal{M}}=\norm{\mathcal{M} - P_\downarrow\mathcal{M}} 
    \leq{2}\sqrt\frac{h}{(1-c)\min\{\Gamma_S,\Gamma_{\overline S}\}},
\end{equation}
by \cref{cor:PionM}. Hence, by the reverse triangle inequality, we have that 
\begin{align}  \label{eq:lem5b}  
    \abs{\norm{\mathcal{M}^\perp {U}\vec\mu} - \norm{\mathcal{M}^\perp {U} P_\downarrow \vec\mu}} &\leq \Vert \mathcal{M}^\perp {U} (\vec\mu-P_\downarrow \vec\mu)\Vert \leq \Vert \mathcal{M}^\perp {U}\Vert \Vert \vec\mu - P_\downarrow\vec\mu\Vert \nonumber \\
    &\leq \Vert \mathcal{M}^\perp {U}\Vert \Vert \mathcal{M}-P_\downarrow \mathcal{M}\Vert 
    \leq {2}\sqrt\frac{h}{(1-c)\min\{\Gamma_S,\Gamma_{\overline S}\}}.
\end{align}
As $\norm{\mathcal{M}^\perp {U}\vec\mu} \geq \norm{\mathcal{M}^\perp {U} P_\downarrow \vec\mu}$ we have from \cref{eq:lem5b} that
\begin{equation}\label{eq:lem5c}
    \norm{\mathcal{M}^\perp {U}\vec\mu} \leq {2}\sqrt\frac{h}{(1-c)\min\{\Gamma_S,\Gamma_{\overline S}\}} + \norm{\mathcal{M}^\perp {U} P_\downarrow \vec\mu},
\end{equation}
so we must bound $\norm{\mathcal{M}^\perp {U} P_\downarrow \vec\mu}$ to get the bound on $\norm{\mathcal{M}^\perp {U}\vec\mu}$ that we want. As $\Vert \mathcal{M}^\perp P_\uparrow {U}P_\downarrow \vec\mu\Vert \leq \Vert \mathcal{M}^\perp {U}P_\downarrow \vec\mu\Vert$, we have by a similar appeal to the reverse triangle inequality as before that
\begin{align*}
    \norm{\mathcal{M}^\perp U P_\downarrow \vec\mu} \leq {2}\sqrt\frac{h}{(1-c)\min\{\Gamma_S,\Gamma_{\overline S}\}}+\norm{\mathcal{M}^\perp P_\uparrow U P_\downarrow \vec\mu}.
\end{align*}
Finally, 
\begin{align*}
    \norm{\mathcal{M}^\perp P_\uparrow U P_\downarrow\vec\mu} \leq \norm{P_\uparrow U P_\downarrow} \leq \epsilon_T,
\end{align*}
by applying the definition of $\epsilon_T$. Combining these results we obtain the desired bound:
\[
    \norm{\mathcal{M}^\perp U\vec\mu} \leq \left( {4}\sqrt\frac{h}{(1-c)\min\{\Gamma_S,\Gamma_{\overline S}\}} + \epsilon_T\right).
\]
Thus, combining everything, we obtain the bound
\begin{align*}
    &\abs{\vec\mu^\dagger {U}^\dagger(1)U'_\perp(1)\vec\mu-1} \leq T \left(4 \sqrt\frac{h}{(1-c)\min\{\Gamma_S,\Gamma_{\overline S}\}} + \epsilon_T\right) {\sqrt{h\kappa}} 
    + T\epsilon_T{\sqrt{h\kappa}}
    + 2hT
    \\&\hspace{6em}= {2hT\left(1 + 
    2\sqrt{\frac{\kappa}{(1-c)\min\{\Gamma_S,\Gamma_{\overline S}\}}}\right) + 2 \epsilon_T T \sqrt{h\kappa} .}
\end{align*}
Plugging this in \cref{eq:temp2}, we arrive at our result.
\end{proof}

\subsection{Perturbed Adiabatic Theorem} \label{s:perturbadiabatic}
In this section, we prove the second component of our main result (\cref{thm:mainresult}).  This error term scales inversely with $T$ and can be understood as deriving an adiabatic theorem that applies locally to $H_S$. 
In particular, we prove the following lemma and associated corollary, which can be used to bound $\epsilon_T$ in \cref{lemma:cheegerperturb}.

\begin{lemma}\label{lemma:perturbadiab} 
\[
    \Vert (U_\perp'(1)-V_{\mathrm{ad}}(1))P_\mu(0)\Vert\leq \max_{s\in[0,1]} \frac{C}{\eta T},
\]
where
$C,\eta$ are as defined in \cref{thm:mainresult}.
\end{lemma}

\begin{cor}\label{cor:local-adiabtic-thm}
\[\Vert \mathcal{M}^\perp(s) U'(s)\vec\mu(0)\Vert \leq \max_{s\in[0,1]}\frac{C}{\eta T}\]
where $C,\eta$ are defined as above,
\end{cor}
\begin{proof} Note that 
\begin{align*}
    \Vert \mathcal{M}^\perp(s) U'(s)\vec\mu(0)\Vert &= \Vert \mathcal{M}^\perp(s) (U'(s)-V_\mathrm{ad}(s))P_\mu(0)\vec\mu(0)\Vert\leq \Vert \mathcal{M}^\perp(s) (U'(s)-V_\mathrm{ad}(s))P_\mu(0)\Vert\\
    &\leq \Vert (U'(s)-V_\mathrm{ad}(s))P_\mu(0)\Vert \leq \max_s \frac{C}{\eta T}.
\end{align*}

\end{proof}

Before presenting the proof of \cref{lemma:perturbadiab}, we must develop a collection of other results. While the analysis becomes somewhat unwieldy, the general techniques are fairly standard  and can be found in a variety of places. Common references for this sort of approach in the standard setting include the lecture notes Ref.~\cite{teufel2003adiabatic,childs2017lecture} and, also, the paper in Ref.~\cite{jansen2007bounds}. While the application of the techniques is somewhat unique, due to the mixture of familiarity and tediousness, we have left some of the proofs to supplemental material.

Before proceeding with this analysis, however, we must consider a technical detail. Recall in the proof of \cref{thm:mainresult} we introduced a block diagonal Hamiltonian $H'$ in \cref{eq:H'}, which is a perturbation of $H(s)$ within the subspace $\overline{S}$ so that for the perturbed Hamiltonian $\mu=\overline{\mu}$. This perturbation was introduced for the purposes of the analysis that follows. Importantly, as we noted, all quantities local to $S$ are invariant under this perturbation (e.g. $P_\mu$, $H_S$, and $V_{\mathrm{ad}}$), as well as many quantities local to $\overline{S}$ (e.g. $P_{\overline{\mu}}$).
With this in mind, to avoid unnecessary notational clutter in what follows, we will only use primed variables for objects that are not invariant under the perturbation. For instance, we will use $H_S$ instead of $H_S'$ as the quantities are equivalent. When acting on states local to $S$ we will use $V_{\mathrm{ad}}$ instead of $V_{\mathrm{ad}}'$. 

 In addition to previously noted references, in this context of perturbed adiabatic theorems

it is also worth noting Ref.~\cite{ohara2008adiabatic}, which derives an adiabatic theorem in the presence of noise. While the general concept of considering an adiabatic theorem under perturbation is similar, the actual formulation of the problem and the bounds they obtain are quite different than ours given their intended applications (i.e. characterizing noisy adiabatic quantum computation). For them, perturbation worsens the bound, whereas in our case, perturbation allows us to provide a tighter bound on the error by understanding the dynamics in terms of local adiabatic evolution.

These details out of the way, we begin by defining a pair of reduced resolvents.
\begin{definition}[Resolvents]\label{def:resolvent}
    Let $H_S:=P_S H P_S= P_S H' P_S$, $H_{\overline{S}}':=P_{\overline{S}} H' P_{\overline{S}}$, and $\mu=\overline{\mu}$ be the corresponding smallest local eigenvalues with associated projectors onto the local ground spaces $P_\mu$ and $P_{\overline{\mu}}$. Let $P_S$ and $P_{\overline{S}}$ be the projectors onto the subspaces $S$ and $\overline{S}$. 
    Define the (reduced) resolvents
    \begin{align*}
    R:&= (H'-\mu I)\vert_{ \ker (H'-\mu I)^\perp}^{-1}(P_\mu+P_{\overline{\mu}})^\perp
    \end{align*}
    and 
    \begin{align*}
    R_\perp:&=(H'+\Cut^\perp - \mu I)\vert_{\ker (H'+\Cut^\perp - \mu I)^\perp}^{-1}(P_\mu+P_{\overline{\mu}})^\perp.
    \end{align*}
\end{definition}
To understand the origin of these resolvents it is helpful to notice that $(H'-\mu I) $ is a map that takes $\mathbb{C}^N\to \mathrm{ker}(H'-\mu I)^\perp\subset \mathbb{C}^N.$ As this map is necessarily not injective, to be able to invert it, we must restrict to a set on which it is. The natural choice for this is $\mathrm{ker}(H'-\mu I)^\perp.$ Thus, we are left with $(H'-\mu I)\vert_{\mathrm{ker}(H'-\mu I)^\perp}:\mathrm{ker}(H'-\mu I)^\perp\to \mathrm{ker}(H'-\mu I)^\perp \subset \mathbb{C}^N,$ which is surjective onto $\mathrm{ker}(H'-\mu I)^\perp$ but not $\mathbb{C}^N.$ Thus, when we invert the map, and get $(H'-\mu I)\vert_{\mathrm{ker}(H'-\mu I)^\perp}^{-1},$ in order to extend this into a map $\mathbb{C}^N\to \mathrm{ker}(H'-\mu I)^\perp\subset\mathbb{C}^N,$ we must choose how it will behave on $\mathrm{ker}(H'-\mu I),$ which we do by precomposing it with $(I-P_\mu-P_{\overline\mu})=(P_\mu+P_{\overline\mu})^\perp$, which takes $\mathbb{C}^N$ to $\mathrm{ker}(H'-\mu I)^\perp.$ Consequently, it holds that 
\begin{equation}\label{eq:resolvent_fact}
    (H'-\mu I) R = R(H'-\mu I)= (I-P_\mu-P_{\overline\mu}).
\end{equation}
From \cref{eq:resolvent_fact}, we can immediately derive several lemmas related to the resolvents defined in \cref{def:resolvent}. To begin we compute the derivative of $R$.
\begin{lemma}\label{lemma:Rderiv}
    Let $R$ be defined as in \cref{def:resolvent}. Then
    \begin{align*}\dot{R} &= -(\dot{P}_{\mu}+\dot{P}_{\overline{\mu}})R - R(\dot{H}'-\dot{\mu}I)R - R(\dot{P}_{\mu}+\dot{P}_{\overline{\mu}}).
    \end{align*}
\end{lemma}
\begin{proof}
Notice first that 
\begin{equation}\label{eq:RonHmuI}
    R(H'-\mu I) = (I-P_\mu-P_{\overline\mu}),
\end{equation}
and 
\begin{equation}\label{eq:RonPmuPmubar}
R(P_\mu+P_{\overline\mu}) = 0.
\end{equation} 

Differentiating \cref{eq:RonHmuI}, we get that 
\begin{align*}
    \dot{R} (H'-\mu I) +R(\dot H'-\dot\mu I)&= -\dot P_\mu - \dot{P}_{\overline{\mu} }
    \implies \dot{R} (H'-\mu I) &= -(\dot P_\mu +\dot P_{\overline\mu})-R(\dot{H}'-\dot{\mu} I).
\end{align*}
Acting from the right by $R$ we obtain
\begin{align}\label{eq:one}
    \dot{R}(I-P_\mu-P_{\overline\mu}) &= -(\dot{P}_\mu + \dot{P}_{\overline{\mu}})R - R(\dot{H}'-\dot{\mu}I)R.
\end{align}
Differentiating \cref{eq:RonPmuPmubar} we get that
\begin{align}\label{eq:two}
    \dot{R} (P_\mu + P_{\overline\mu}) &= - R \left(\dot{P}_\mu + \dot{P}_{\overline\mu}\right).
\end{align}
Adding \cref{eq:one} and \cref{eq:two} we get our result.
\end{proof}

To obtain our final bound it is essential to have bounds on the norms of a number of different quantities containing the resolvents. These bounds are established in the following lemma. 
\begin{lemma}\label{prop:Rprop}
    Let the reduced resolvent $R$ be defined as in \cref{def:resolvent}. Then, the following hold:
    \begin{enumerate}
    \item $\Vert R\Vert \leq \frac{1}{\min\{\Gamma_S,\Gamma_{\overline{S}}\}}$ \label{propRitem1}
    \item $\Vert \dot{R}\Vert \leq 4\max\left\{\frac{\Vert \dot{H}_S\Vert}{\Gamma_S^2}\frac{\Vert \dot{H}_{\overline{S}}'\Vert}{\Gamma_{\overline{S}}^2}\right\}$, \label{propRitem2}
    \item $\Vert \dot{R}P_S\Vert\leq 4\frac{\Vert \dot{H}_S\Vert}{\Gamma_S^2}$, and \label{propRitem3}
    \item $\Vert \dot{R}P_{\overline{S}}\Vert\leq 4\frac{\Vert \dot{H}_{\overline{S}}'\Vert}{\Gamma_{\overline{S}}^2}.$ \label{propRitem4}
    \end{enumerate}
\end{lemma}
\begin{proof}
    \cref{propRitem1} follows by simply placing a norm on \cref{eq:resolvent_fact} and applying the Cauchy-Schwarz inequality:
    \[
    \norm{R}\leq \frac{\norm{I-P_\mu-P_{\overline{\mu}}}}{\norm{H'-\mu I}} \leq \frac{1}{\min\{\Gamma_S,\Gamma_{\overline{S}}\}}
    \]    
    \cref{propRitem2} comes from placing norms on the result of \cref{lemma:Rderiv}:
    \begin{align*}
        \Vert \dot R\Vert &= \Vert (\dot{P}_{\mu}+\dot{P}_{\overline{\mu}})R + R(\dot{H}'-\dot{\mu}I)R + R(\dot{P}_{\mu}+\dot{P}_{\overline{\mu}})\Vert\\
        &\leq \max\{ \Vert P_S(\dot{P}_\mu R + R(\dot H'-\dot \mu I)R + R\dot P_\mu)P_S\Vert, \Vert P_{\overline S}(\dot{P}_{\overline\mu} R + R(\dot H'-\dot \mu I)R + R\dot P_{\overline{\mu}})P_{\overline S}\Vert\}\\
        &\leq \max\{ 2\Vert \dot{P}_\mu\Vert \Vert R P_S\Vert + \Vert RP_S\Vert^2\Vert(\dot H'-\dot \mu I)P_S \Vert, 2\Vert \dot{P}_{\overline\mu}\Vert \Vert RP_{\overline S}\Vert + \Vert RP_{\overline S}\Vert^2 \Vert (\dot H'-\dot \mu I) P_{\overline S}\Vert  \}.
    \end{align*}
    Using the facts that $\Vert \dot P_\mu\Vert \leq \Vert \dot{H}_S\Vert/\Gamma_S$ (\cref{prop:Pmunorms}), $\norm{RP_S}\leq 1/\Gamma_S$, $\norm{RP_{\overline{S}}}\leq 1/\Gamma_{\overline{S}}$, $\norm{(\dot H'-\dot\mu I) P_S}\leq 2\norm{\dot H_S}$, and $\norm{(\dot H'-\dot\mu I) P_{\overline{S}}}\leq 2\norm{\dot H_{\overline{S}}'}$ completes the proof:
    \begin{align*}
         \Vert \dot R\Vert &\leq \max\left\{2\frac{\Vert\dot H_S\Vert}{\Gamma_S} \frac{1}{\Gamma_S} + \frac{2\Vert\dot{H}_S\Vert}{\Gamma_S^2}, 2\frac{\Vert\dot H_{\overline S}'\Vert}{\Gamma_{\overline S}} \frac{1}{\Gamma_{\overline S}} + \frac{2\Vert\dot{H}_{\overline S}'\Vert}{\Gamma_{\overline S}^2}\right\} \leq 4\max\left\{\frac{\Vert \dot{H}_S\Vert}{\Gamma_S^2},\frac{\Vert \dot{H}_{\overline{S}}'\Vert}{\Gamma_{\overline{S}}^2}\right\}
    \end{align*}
    \cref{propRitem3} and \cref{propRitem4} follow by almost identical arguments, except the projectors eliminate one of the terms in the maximization.
\end{proof}

We now note a fact about perturbations of resolvents, which follows from the second resolvent identity (e.g. Thm. 4.8.2 of Ref.~\cite{hille1996functional}), which states that $R-R_\perp  = R\Delta^\perp R_\perp = -R_\perp \Delta^\perp R$. Defining
\begin{equation}\label{fact:ResolventIdentity}
    L := (I+R\Cut^\perp),
\end{equation}
and rearranging the second resolvent identity lets us note that 
\begin{equation}
R_\perp = L^{-1}R.
\end{equation}
Furthermore, differentiating this equation yields that
\begin{equation}
    \dot{R}_\perp = L^{-1}\dot{R} - L^{-1}(\dot{R}\Cut^\perp+R\dot{\Cut}^\perp)L^{-1}R. 
\end{equation}

These useful facts established, we now present the following proposition, which bounds the norm of the operator $L$.
\begin{prop}\label{prop:normLinv}
    Let $L$ be defined as above in \cref{fact:ResolventIdentity}. Then 
    \[\Vert L^{-1} \Vert \leq \frac{1}{\eta},\] where $\eta := 1-\Vert \Cut^\perp\Vert/\min\{\Gamma_S,\Gamma_{\overline{S}}\}.$
\end{prop}
\begin{proof}
    Consider a Neumann series 
    \[
    L^{-1}=\sum_{j=0}^\infty (I-L)^j=\sum_{j=0}^\infty  (R\Delta^\perp)^j.
    \]
    Taking the norm of this expression and using \cref{assmpn:tbd2} we have that
    \begin{align*}
    \norm{L^{-1}} &\leq \sum_{j=0}^\infty \norm{R\Delta^\perp}^j\leq \sum_{j=0}^\infty \norm{R}^j\norm{\Delta^\perp}^j = \frac{1}{1-\Vert R\Vert \Vert \Cut^\perp \Vert} \leq \frac{1}{1-\frac{\Vert\Cut^\perp\Vert}{\min\{\Gamma_S,\Gamma_{\overline{S}}\}}}.
    \end{align*}

\end{proof}

To proceed further, we now define the operator
\begin{equation}\label{def:F}
   F:=R_\perp \dot{P}_{\mu}P_{\mu}+P_{\mu}\dot{P}_{\mu}R_\perp.
\end{equation}

Leaving the details of the proofs to the supplemental material as they are not terribly enlightening and consist mostly of tedious algebra we state the following technical lemma about this operator:

\begin{lemma} \label{lemma:UperpPmuPmubarUperp}
Let $U_\perp',$ ${P}_{\mu},$ and $F$ be as defined above. Then, 
    \[(U_\perp')^{\dagger}[\dot{{P}}_{\mu},{P}_{\mu}]U_\perp' = \frac{1}{iT}\left(\dot{\tilde{F}}-(U_\perp')^{\dagger}\dot{F}U_\perp'\right),\] where $\tilde{F}:=(U_\perp')^{\dagger} FU_\perp'.$
\end{lemma}

With this result, at last, we are ready to move towards our ultimate goal of bounding the difference between the action of the unitaries $U_\perp'$ and $V_{\mathrm{ad}}$ acting on the initial ground state $P_\mu(0)$ of the space $S$. In the lemma that follows, using standard techniques for proving the usual adiabatic theorem, we will express this difference as an integral expression in terms of the operator $F$ defined above. The result of this manipulations will then be able to be easily bounded with the help of the reduced resolvents described previously. 

\begin{lemma}\label{lemma:perturbedwonorms} Let $U_\perp', V_{\mathrm{ad}}, F,$ and $ P_\mu$ be defined as above. Then, 
    \begin{align*}
    (U_\perp'(1) - V_{\mathrm{ad}}(1))&P_\mu(0)=\nonumber \\&-\frac{1}{iT}U_\perp'(1)\left[[(U_\perp')^{\dagger}{F}{P}_{\mu} V_{\mathrm{ad}}]_0^1-\int_0^1\mathrm{d}s\left((U_\perp')^\dagger{F} \dot{{P}}_{\mu}{P}_{\mu}V_{\mathrm{ad}}+(U_\perp')^\dagger\dot{F}{P}_{\mu} V_{\mathrm{ad}}\right)\right].
    \end{align*}
\end{lemma}
\begin{proof} We rewrite the left hand side as follows:
    \begin{align*}
    (U_\perp'(1)-V_{\mathrm{ad}}(1))P_\mu(0) &= U_\perp'(1)(I-(U_\perp')^{\dagger}(1)V_{\mathrm{ad}}(1))P_\mu(0)\\
    &= -U_\perp'(1)\int_{0}^1 \mathrm{d}s (U_\perp')^{\dagger}(iT(H'+\Cut^\perp - H) + [\dot{P}_\mu,P_\mu] )V_{\mathrm{ad}}P_\mu(0)\\
    &= -U_\perp'(1)\int_{0}^1 \mathrm{d}s (U_\perp')^{\dagger} [\dot{P}_\mu,P_\mu]V_{\mathrm{ad}}P_\mu(0).
\end{align*}
Using \cref{lemma:UperpPmuPmubarUperp}, we replace $(U_\perp')^{\dagger} [\dot{P}_\mu, P_\mu]$ with $\frac{1}{iT}\left(\dot{\tilde{F}}-(U_\perp')^\dagger \dot{F} U_\perp'\right)(U_\perp')^\dagger$ and integrate the first term by parts to get 
\begin{align*}
    &(U_\perp'(1)-V_{\mathrm{ad}}(1))P_\mu(0) = -\frac{1}{iT}U_\perp'(1)\int_{0}^1 \mathrm{d}s \left(\dot{\tilde{F}} - (U_\perp')^\dagger \dot{F}U_\perp'\right)(U_\perp')^\dagger V_{\mathrm{ad}}P_\mu(0)\\
    &\quad= -\frac{1}{iT} U_\perp'(1) \left(\left[ \tilde{F}(U_\perp')^\dagger V_{\mathrm{ad}}P_\mu(0)\right]_0^1 -\int_0^1\mathrm{d}s (U_\perp')^\dagger \left(F\left(iT\Cut^\perp + [\dot{P}_\mu,P_\mu]\right)+\dot{F}\right)V_{\mathrm{ad}}P_\mu(0)\right)\\
    &\quad= -\frac{1}{iT}U_\perp'(1)\left[[(U_\perp')^\dagger{F}{P}_{\mu} V_{\mathrm{ad}}]_0^1-\int_0^1\mathrm{d}s\left((U_\perp')^\dagger{F} \dot{{P}}_{\mu}{P}_{\mu}V_{\mathrm{ad}}+(U_\perp')^\dagger\dot{F}{P}_{\mu} V_{\mathrm{ad}}\right)\right].
\end{align*}
\end{proof}

To obtain our desired result we simply need to place norms on the expression derived in \cref{lemma:perturbedwonorms}. Before doing so, however, we need bounds on a few simpler quantities that show up in the integrand of this expression. These bounds are derived in the following lemma.

\begin{lemma}\label{cor:Fbounds}
Let all quantities be defined as above.
    Then
    \[\max\{\Vert FP_\mu\Vert, \Vert F(P_S-P_\mu)\Vert\}\leq \frac{1}{\eta}\frac{\Vert \dot{H}_S\Vert}{\Gamma_S^2} \]
    and $\Vert\dot{F}P_\mu \Vert$ is upper bounded by
    \begin{align*}
    \frac{1}{\eta}\left(\frac{1}{\eta}\frac{\Vert\dot{H}_S\Vert}{\Gamma_S^2}\left(\frac{\Vert\dot{\Cut}^\perp\Vert}{\min\{\Gamma_S,\Gamma_{\overline{S}}\}}+4\Vert \Cut^\perp\Vert \max\left\{\frac{\Vert \dot{H}_S\Vert}{\Gamma_S^2},\frac{2\Vert\dot{H}_{\overline{S}}\Vert+\Vert \dot{H}_{S}\Vert}{\Gamma_{\overline{S}}^2}\right\}\right)+9\frac{\Vert\dot{H}_S\Vert^2}{\Gamma_S^3}+\frac{\Vert\ddot{H}_S\Vert}{\Gamma_S^2}\right).
    \end{align*}
\end{lemma}
\begin{proof}
    First, note that $FP_\mu = L^{-1}R\dot P_\mu P_\mu$ and $F(P_S-P_\mu) = P_\mu \dot{P}_\mu L^{-1}RP_S.$
    Applying norms and using that $\Vert \dot P_\mu\Vert \leq \Vert \dot{H}_S\Vert/\Gamma_S$ \cite[Lemma 8]{jansen2007bounds}, \cref{prop:Rprop,prop:normLinv}, we get our first bound. For the next bound, we recall that 
     \begin{align}\label{eq:dotFP}
     \dot{F}P_\mu &=-L^{-1}\left(\dot\Cut^\perp R+\Cut^\perp\dot{R} \right)L^{-1}R\dot{P}_\mu P_\mu + L^{-1}\dot{R}\dot{P}_\mu P_\mu + L^{-1}{R}\ddot{P}_\mu P_\mu  - P_\mu\dot{P}_\mu L^{-1}{R}\dot{P}_\mu \nonumber\\
     &=: A_0+A_1+A_2+A_3, 
     \end{align} 
     where $A_n$ is the $n^{\mathrm{th}}$ term in the sum. Then using that $\Vert (I-P_\mu)\ddot{P}_\mu P_\mu\Vert \leq \frac{\Vert\ddot{H}_S\Vert}{\Gamma_S}+4\frac{\Vert\dot{H}_S\Vert^2}{\Gamma_S^2}$ \cite[Lemma 8]{jansen2007bounds}, we can bound the norms of the individual $A_n$. In particular, we have that 
     \begin{align*}
         \Vert A_0\Vert &= \Vert L^{-1}(\dot{\Cut}^\perp R + \Cut^\perp \dot R)L^{-1} R\dot{P}_\mu P_\mu\Vert =  \Vert L^{-1}(\dot{\Cut}^\perp R + \Cut^\perp \dot R)L^{-1} R(P_S-P_\mu)\dot{P}_\mu P_\mu\Vert\\
         &\leq  \Vert L^{-1}\Vert^2 \Vert (\dot{\Cut}^\perp R + \Cut^\perp \dot R)\Vert  \Vert R(P_S-P_\mu)\Vert \Vert \dot{P}_\mu\Vert  \\
         &\leq \frac{1}{\eta^2}\frac{\Vert \dot H_S\Vert }{\Gamma_S^2} \left(\frac{\Vert\dot{\Cut}^\perp\Vert}{\min\{\Gamma_S,\Gamma_{\overline{S}}\}}+4\Vert \Cut^\perp\Vert \max\left\{\frac{\Vert \dot{H}_S\Vert}{\Gamma_S^2},\frac{\Vert\dot{H}_{\overline{S}}'\Vert}{\Gamma_{\overline{S}}^2}\right\}\right).
    \end{align*}
    Bounding the other terms is more straightforward:
    \begin{align*}
         \Vert A_1\Vert &= \Vert L^{-1} \dot{R} \dot{P}_\mu \Vert \leq \Vert L^{-1}\Vert \Vert \dot R P_S\Vert \Vert \dot P_\mu\Vert \leq 4\frac{1}{\eta}\frac{\Vert \dot H_S\Vert^2}{\Gamma_S^3},\\        
         \Vert A_2\Vert &= \Vert L^{-1} R\ddot P_\mu P_\mu\Vert = \Vert L^{-1} R P_S(I-P_\mu)\ddot P_\mu P_\mu\Vert \leq \Vert L^{-1} \Vert\Vert RP_S\Vert \Vert(I-P_\mu)\ddot P_\mu P_\mu\Vert \\&\leq \frac{1}{\eta}\left(\frac{\Vert \ddot{H}_S\Vert }{\Gamma_S^2}+4\frac{\Vert \dot H_S\Vert^2}{\Gamma_S^3}\right), \\
         \Vert A_3\Vert &= \Vert P_\mu \dot{P}_\mu L^{-1} R \dot{P}_\mu\Vert \leq \Vert L^{-1}\Vert \Vert \dot{P}_\mu\Vert^2\Vert RP_S\Vert \leq \frac{1}{\eta} \frac{\Vert \dot{H}_S\Vert^2}{\Gamma_S^3}.
     \end{align*}
     Using these results and taking the norm of \cref{eq:dotFP} we obtain via the triangle inequality that
     \begin{align*}
         \Vert \dot FP_\mu\Vert &\leq \Vert A_0\Vert +\Vert A_1\Vert +\Vert A_2\Vert +\Vert A_3\Vert\\
         &\leq \left(\frac{1}{\eta^2}\frac{\Vert \dot H_S\Vert }{\Gamma_S^2} \left(\frac{\Vert\dot{\Cut}^\perp\Vert}{\min\{\Gamma_S,\Gamma_{\overline{S}}\}}+4\Vert \Cut^\perp\Vert \max\left\{\frac{\Vert \dot{H}_S\Vert}{\Gamma_S^2},\frac{\Vert\dot{H}_{\overline{S}}'\Vert}{\Gamma_{\overline{S}}^2}\right\}\right)\right)+\left(4\frac{1}{\eta}\frac{\Vert \dot H_S\Vert^2}{\Gamma_S^3}\right)\\
         &\qquad+\left(\frac{1}{\eta}\left(\frac{\Vert \ddot{H}_S\Vert }{\Gamma_S^2}+4\frac{\Vert \dot H_S\Vert^2}{\Gamma_S^3}\right)\right)+\left(\frac{1}{\eta} \frac{\Vert \dot{H}_S\Vert^2}{\Gamma_S^3}\right)\\
         &\leq \frac{1}{\eta}\left(\frac{1}{\eta}\frac{\Vert\dot{H}_S\Vert}{\Gamma_S^2}\left(\frac{\Vert\dot{\Cut}^\perp\Vert}{\min\{\Gamma_S,\Gamma_{\overline{S}}\}}+4\Vert \Cut^\perp\Vert \max\left\{\frac{\Vert \dot{H}_S\Vert}{\Gamma_S^2},\frac{\Vert\dot{H}_{\overline{S}}'\Vert}{\Gamma_{\overline{S}}^2}\right\}\right)+9\frac{\Vert\dot{H}_S\Vert^2}{\Gamma_S^3}+\frac{\Vert\ddot{H}_S\Vert}{\Gamma_S^2}\right).
     \end{align*}
      Finally, we seek to eliminate the remaining term that depends on the fact that we are considering the perturbed Hamiltonian $H'$---that is $\Vert H_{\overline{S}}'\Vert$. By definition $\dot{H}_{\overline{S}}' = \dot{H}_{\overline{S}}+(\dot{\mu}-\dot{\overline{\mu}})P_{\overline{S}}$. By the Hellman-Feynman theorem, we find that $\Vert\dot{H}_{\overline{S}}'\Vert = 2\Vert\dot{H}_{\overline{S}}\Vert+\Vert\dot{H}_{S}\Vert$, completing the proof. 
\end{proof}

At last we have assembled all of the necessary ingredients to prove our target result (\cref{lemma:perturbadiab}). The proof is now straightforward and is presented below.
\begin{proof}[Proof of \cref{lemma:perturbadiab}]
    Using \cref{lemma:perturbedwonorms} we can write $D:=\Vert(U'_\perp(1)-V_{\mathrm{ad}}(1))P_\mu(0)\Vert$ as 
    \[D=\bigg\Vert-\frac{1}{iT}U_\perp'(1)\left[[(U_\perp')^\dagger{F}{P}_{\mu} V_{\mathrm{ad}}]_0^1-\int_0^1\mathrm{d}s\left((U_\perp')^\dagger{F} \dot{{P}}_{\mu}{P}_{\mu}V_{\mathrm{ad}}+(U_\perp')^\dagger\dot{F}{P}_{\mu} V_{\mathrm{ad}}\right)\right]\bigg\Vert.\]
    We obtain
    \begin{align*}
        D&\leq \frac{1}{T}\left(\Vert F(0)P_\mu(0)\Vert+\Vert F(1)P_\mu(1)\Vert+\int_{0}^1\mathrm{d}s\left(\Vert F\dot P_\mu P_\mu\Vert + \Vert \dot F P_\mu\Vert \right)\right) \\
        &\leq \frac{1}{T}\left( 2\max_{s\in [0,1]}\Vert F(s)P_\mu(s)\Vert + \max_{s\in[0,1]}\Vert F(P_S-P_\mu)\Vert \Vert \dot P_\mu\Vert + \max_{s\in[0,1]}\Vert \dot F P_\mu\Vert \right).
    \end{align*}
    Each of these three terms can be directly bounded via \cref{cor:Fbounds}. Summing up these bounds and including the maximization over $s$ yields the result. 
\end{proof}

\section{Stoquastic Hamiltonians} \label{s:stoquasticH}
So-called stoquastic Hamiltonians are frequently studied in the context of adiabatic quantum computing. Typically defined as Hamiltonians $H$ with all real, non-positive off-diagonal matrix elements, stoquastic Hamiltonians do not present a sign problem to Monte Carlo simulations.
In the math literature, such matrices have numerous names, such as (Hermitian) $Z$-matrices, negated Metzler matrices, or essentially nonpositive matrices. It is worth noting that a broader class of Hamiltonians are still sign-problem free. This broader class, which consists of any \textit{balanced} Hamiltonian, those with frustration index zero, was proposed to be a better definition of stoquastic in Ref.~\cite{jarret2018hamiltonian}. The same concept has been referred to as Hamiltonians of Vanishing Geometric Phase (VGP) form in Ref.~\cite{Hen2021}. 

Putting aside the confusion that arises from an unfortunate proliferation of definitions, it can unambiguously be stated that for balanced Hamiltonians, an appropriate choice of $S$ allows the Cheeger ratio that appears in our bounds can be both upper- and lower-bounded by the global spectral gap $\gamma$. This allows us to re-express our results in terms of this more directly physical quantity. In particular, the following holds for balanced Hamiltonians ~\cite{jarret2018hamiltonian}:
\begin{equation}\label{eq:stoq_bnd}
    \tilde{\gamma}:= \sqrt{\gamma(2Q+\gamma)} \geq \min_{S\subset V} h
\end{equation}
where $\gamma$ is the (global) spectral gap of $H+\Delta$ and $Q=\max_i \sum_{j\neq i} |H_{ij}|.$ We implicitly defined the upper bound in \cref{eq:stoq_bnd} as $\tilde{\gamma}$, which we note scales like the global gap $\sqrt{\gamma}$ for sparse Hamiltonians (as naturally arise in physical systems). 
As an immediate consequence of \cref{eq:stoq_bnd} we have the following corollary to \cref{thm:mainresult} that applies to stoquastic Hamiltonians when $S$ is taken as the minimum cut. 
\begin{cor}\label{corr:stoq}
    Let $H+\Delta$ be a balanced, one parameter, twice-differentiable family of stoquastic Hamiltonians respecting the grading $\mathcal{H}=S\oplus\overline{S}$ and $\Delta(s)$ be a block-antidiagonal perturbation, where $s:=t/T\in[0,1]$ and for all $s$, $S=\argmin_{S'} h_{S'}$. Then, the difference between the unitary evolution $U(1)$ generated by the Hamiltonian $H(s)+\Cut(s)$ acting for time $T$ and the adiabatic unitary evolution $V_{\mathrm{ad}}(1)$ generated by the block-diagonal Hamiltonian $H(s)$ is bounded as 
    \[
    \norm{(U(1)-V_{\mathrm{ad}}(1))P_{\mu}(0)}\leq \max_{s\in[0,1]} \left[B\sqrt{\tilde\gamma T}+\frac{C}{\eta T}\right], 
    \]
    where $\tilde\gamma$ is defined as in \cref{eq:stoq_bnd}, 
\begin{align*}
B=2\sqrt{1 +  2\sqrt\frac{\kappa}{(1-c)\min\{\Gamma_S,\Gamma_{\overline S}\}} + \sqrt{\frac{\kappa}{\gamma}}\epsilon_T}
\end{align*}
and $C, \epsilon_T, \eta, \kappa,$ and $c$ are as defined in \cref{thm:mainresult}. 

\end{cor}
\begin{proof}
    Simply substitute the bound in \cref{eq:stoq_bnd} into \cref{thm:mainresult}.
\end{proof}
We expect that stronger bounds than \cref{corr:stoq} are possible for stoquastic Hamiltonians as nowhere in our derivation of \cref{thm:mainresult} did we make use of the special properties of such Hamiltonians. In particular, stoquastic Hamiltonians have non-negative ground states, which could potentially be leveraged to prove a stronger version of \cref{lemma:cheegerperturb} .

\section{Examples}\label{s:examples}
In this  section, to match the notation of the physics literature we use standard Dirac (``bra-ket'') notation. We begin by describing a toy example based on a so-called $k$-barbell Hamiltonian that most starkly demonstrates the regimes where our bounds are useful. We then analyze the classical adiabatic Grover search problem~\cite{RolandPRA2002} through the lens of our new results.

\subsection{Barbell Graph}\label{s:barbell-graph}
Our results are most useful in Hamiltonians that induce a ``bottlenecked'' ground state. Although bottlenecked ground states do not require bottlenecked Hamiltonians, they are easiest to understand in such cases and one of the most extreme ones is a $k$-barbell Hamiltonian. This example may seem somewhat contrived physically, but within the context of exploiting quantum dynamics to perform search tasks, cases like these are expected.

Specifically, we construct a Hamiltonian $H=D-A$ whose nonzero off-diagonal terms are coincident with the nonzero terms of the adjacency matrix $A$ of a lopsided, weighted $k$-barbell graph $G(n,k)$ and $D$ is a diagonal matrix specified below. This example consists of two cliques $K_n$ and $K_{3n/4}$ with disjoint vertex sets of $n$ and $\lfloor3n/4\rfloor$ vertices respectively, connected by a weighted path graph $P_k$ of $k$ vertices. All vertex sets $V(K_n),V(K_{3n/4}),$ and $P_k$ are disjoint and there exists a unique edge $e_0$ connecting one end of $P_k$ to $K_n$ and similarly a unique edge $e_1$ connecting the other end of $P_k$ to $K_{3n/4}$. (See \cref{fig:barbell-graph-example}a.) We will take $k$ to be even. The weight function will assign weights $w$ to any edge in the path and weight $1$ elsewhere.
In particular, we have that
\[
    A_{ij} = \begin{cases}
        1 & \text{if $\{i,j\} \in E(K_n)$ or $\{i,j\} \in E(K_{3n/4})$}
        \\w & \text{if $\{i,j\} \in E(P_k)$ or $\{i,j\} = e_0,e_1$}
        \\0 & \text{otherwise.}
    \end{cases}
\]
Because we wish to construct a time-dependent Hamiltonian where adiabatic theory might apply, we consider a simple case where we adjust the diagonal term of marked standard basis states $\ket{m_S} \in V\left(K_n \right)$ and $\ket{m_{\overline{S}}} \in V\left(K_{3n/4} \right)$ where we will linearly decrease the size of the diagonal term. That is, we choose 
\[
    D_{ij} = \begin{cases}
        -s & \text{if $i=j=m_S,m_{\overline{S}}$}
        \\0 & \text{otherwise.}
    \end{cases}
\]
Now, consider the $S$ that cuts the original graph into two identical halves by separating the graph halfway along the path as in \cref{fig:barbell}. $\Delta$ contains only the matrix elements corresponding to the single cut edge in the center of $P_k$ and $\norm{\Delta} = w$. 

The Cheeger ratio $h\in\mathcal{O}(n^{-1}e^{-k})$ bounds the leakage from the local ground state of $S$ and also upper bounds the global gap $\gamma$. Furthermore, $\min_s \Gamma_S=\min_s\Gamma_{\overline{S}}\gg \gamma$. The lopsidedness of the barbell precludes obvious symmetry arguments, though we would anticipate that something like local adiabatic dynamics should govern evolution.

We numerically evaluate our bound for this example with $n=100$, $k=6$, and $w=5$. For simplicity, we assume that $\epsilon_T$ is consistent with ``Theorem''~\ref{sthm:folk_adiabatic}, as the particular adiabatic theorem does not matter for the sake of our comparison. As shown in \cref{fig:barbell-graph-example} our bounds yield the expected behavior: they are effective at ``intermediate'' timescales where neither the local adiabatic error, nor the leakage error are too large. This timescale is much shorter than that specified by the usual adiabatic theorem. 

Also shown in \cref{fig:barbell-graph-example}b, we compare our bound to a naive perturbation theory bound where one applies standard time-dependent perturbation theory to $H+\Delta$ to remove $\Delta$, followed by the standard adiabatic theorem to $H_S$. As we emphasized in the motivating example of the time-independent case this approach suffers from yielding an error term that scales with $\norm{\Delta}$ that, in contrast to our bound, is not suppressed by any factors of $T$. In general, as in this example, $\norm{\Delta}$ can be much larger than $h$, yielding a large gap between the na\'ive bound and our bound. This also leads us to our next example, in which the analysis provides lower bounds to this example equally well with a slight modification to account for the local ``tails.'' 

\begin{figure}\label{fig:barbell}
\centering\includegraphics[width=0.7\textwidth]{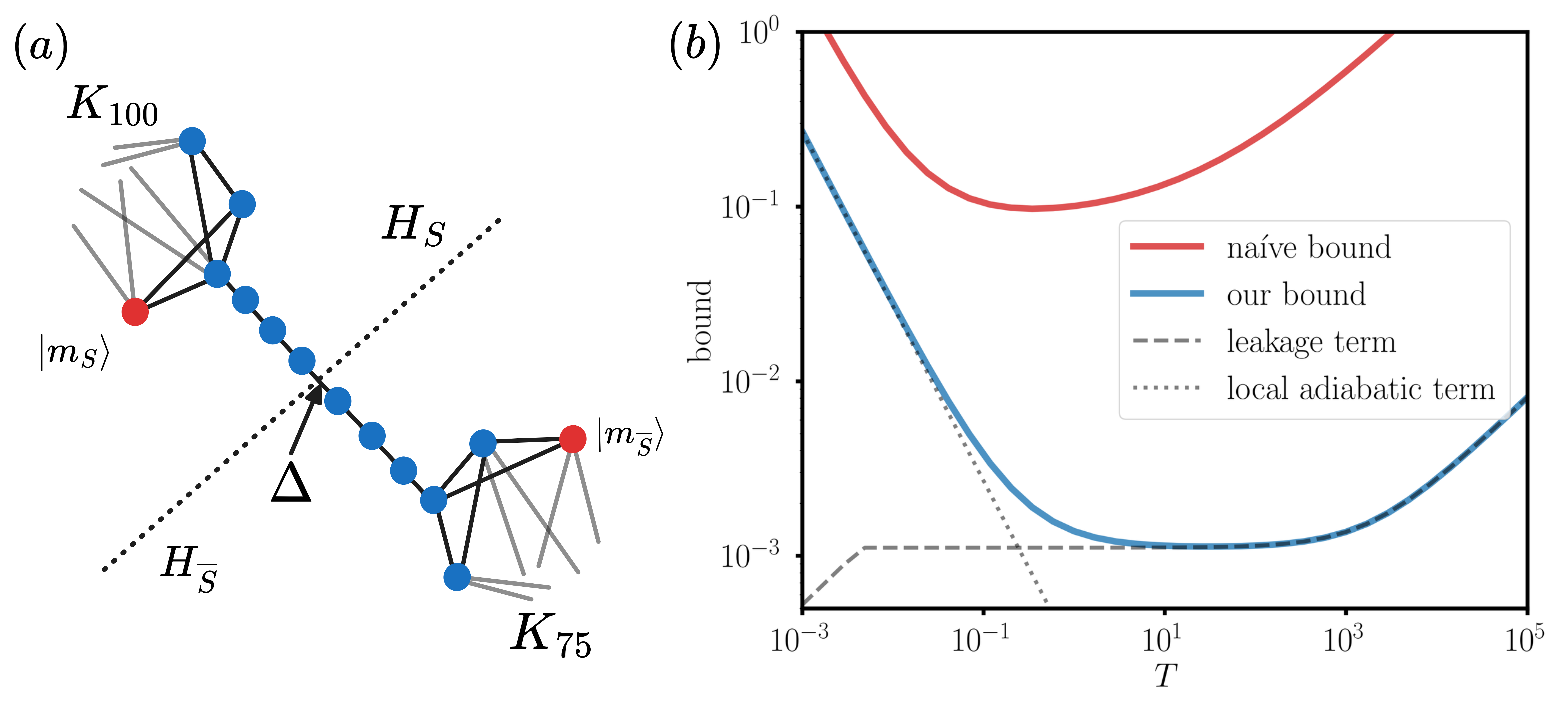}
\caption{Our bounds for the barbell graph example in \cref{s:barbell-graph}. Starting with a Hamiltonian associated with a lopsided, weighted barbell graph as described in the main text with parameters $n=100$, $k=6$, and $w=5$, we linearly turn off a potential on the two marked states $\ket{m_S}, \ket{m_{\overline{S}}}$ over a time $T$. We evaluate our bound as a function of $T$ and, as expected, observe an ``intermediate'' timescale region where neither leakage errors nor local adiabatic errors are overly large. This timescale is significantly less than the standard adiabatic timescale applied to the ground state of this system---to reach the rigorous error bound of $10^{-2}$ obtained by our results the standard adiabatic theorem applied to the full Hamiltonian would specify an interpolation time $T$ vastly larger than our theorem. We also compare our bound to the na\'ive perturbation theory-based bound (shown in red), highlighting the value of our approach. }\label{fig:barbell-graph-example}
\end{figure}

\subsection{Adiabatic Grover Search}\label{s:grover}
As a  second  simple application, we demonstrate how our bounds allow us to analyze analog Grover search beyond the adiabatic regime. In this problem, we seek to reach a so-called marked state---a computational basis state $\ket{m}$ for $m\in\{0,1\}^n$, where $n$ is the number of qubits. To do so, we perform quantum annealing starting from the ground state $\ket{\psi}$ of an initial Hamiltonian $H_0=\id-\ket{\psi}\bra{\psi},$
where $\ket{\psi}:=2^{-n/2}\sum_{w\in\{0,1\}^n}\ket{w}$ is the even superposition over all bit strings, and seek to prepare the ground state of the final Hamiltonian $
H_1=\id-\ket{m}\bra{m}.$

In usual treatments of this problem~\cite{RolandPRA2002,jansen2007bounds}, one considers a general annealing schedule that interpolates the Hamiltonian from $H_0$ to $H_1$ as a function of $s:=\frac{t}{T}\in[0,1]$. Here, for technical reasons that make the analysis slightly easier we will consider the reverse (but, equivalent) problem of interpolating from $H_1$ to $H_0$ to prepare the state $\ket{\psi}$, starting from $\ket{m}$. Our analysis will proceed similarly to that of Ref.~\cite{RolandPRA2002} for an entirely adiabatic schedule. In particular, consider an annealing schedule specified by a function $f(s)$ that satisfies the differential equation $\dot{f}(s)=kg^p(f(s)) $
with $f(0)=0$, where $p\in(1,2)$ and $k:=\int_0^1 du\, g^{-p}(u)$ are constants. Letting $N:=2^{n}$, 
\begin{equation}\label{eq:gap}
g(\tilde s):=\sqrt{1-4\cdot \frac{N-1}{N}\tilde{s}(1-\tilde{s})},
\end{equation}
is the eigenvalue gap for the corresponding time-dependent Hamiltonian given by
\begin{align}\label{eq:Hgrover}
&H(s)=(1-f(s))H_1+f(s)H_0\nonumber \\
&=I-\left(1-\frac{f(s)(N-1)}{N}\right)\ket{m}\bra{m}-\frac{f(s)(N-1)}{N}\ket{\perp}\bra{\perp} -\frac{f(s)\sqrt{N-1}}{N}\left(\ket{\perp}\bra{m}+\ket{m}\bra{\perp}\right),
\end{align}
where $\ket{\perp}:=\frac{1}{\sqrt{N-1}}\sum_{j\neq m}\ket{j}$. The second line indicates that the dynamics in this problem are confined to a two dimensional subspace given by $\mathrm{Span}(\{\ket{m}, \ket{\perp}\})$. 

To analyze this problem for faster-than-adiabatic schedules, we consider (1) applying the standard adiabatic theorem for $s=0$ to some $s_-$, (2) applying our theorem for $s=s_-$ to some $s_+$, where the standard adiabatic theorem does not hold. This allows us to rigorously understand leakage from the subspace $S$ given by the marked state $\ket{m}\bra{m}$ to $\overline{S}$ (which due to the effective two dimensional dynamics is just $\ket{\perp}\bra{\perp}$). Note that we must only bound the errors arising due to leakage, not due to transitions to higher excited states in the subspace $S$. This significant simplification is because there is a single state in the subspace $S$. This feature of the analog Grover search problem makes it a mathematically convenient toy scenario to see how our bounds allow one to understand annealing beyond the adiabatic regime. 

For step (1), we simply bound the standard adiabatic error.  For this purpose we use Theorem 3 of Ref.~\cite{jansen2007bounds}, but we note that very similar statements can be found in earlier works, see i.e. the lecture notes in Ref.~\cite{teufel2003adiabatic}.
The theorem is

restated here in a less general form for convenience:

\begin{theorem}[Theorem 3 of Ref.~\cite{jansen2007bounds}]\label{thm:JRS}
    Suppose the ground state of a Hamiltonian H(s) is separated by a gap $g(s)$ from the rest of the spectrum of $H(s)$ and $H, \dot{H},$ and $\ddot{H}$ are bounded operators. Then, the adiabatic error between the projector $P_T(s)$ onto the state during an adiabatic anneal of time $T$ and the projector $P(s)$ onto the ground state of $H(s)$ is bounded as $\norm{P_T(s)-P(s)}\leq A(s)$ where
    \[
    A(s)=\frac{1}{T}\, \frac{\norm*{\dot{H}}}{g^2(f(\tau))}\Bigg |_{\mathrm{u.b.}}+\frac{1}{T}\int_0^s\left(\frac{\norm*{\ddot{H}}}{g^2(f(\tau))}+7\frac{\norm*{\dot{H}}^2}{g^3(f(\tau))}\right)d\tau,
    \]
    and $a|_{\mathrm{u.b.}}:=a(\tau=0)+a(\tau=s)$.
\end{theorem}

To apply the theorem, we differentiate \cref{eq:Hgrover} to yield\footnote{We note that Ref.~\cite{jansen2007bounds} was missing an (ultimately inconsequential) factor of $p$ in $\ddot{H}$ in their analysis of this problem.}
\begin{align}
&\dot{H}=k g^p(f(s))(\ket{m}\bra{m}-\ket{\psi}\bra{\psi}) &\implies &&\norm*{\dot{H}}\leq kg^p(f(s)), \\
&\ddot{H}=p k^2 g^{2p-1}(f(s))\dot{g}(f(s))(\ket{m}\bra{m}-\ket{\psi}\bra{\psi}) &\implies &&\norm*{\ddot{H}}\leq pk^2g^{2p-1}(f(s))\dot{g}(f(s)),
\end{align}
where we use that the norm of the difference of two projectors is upper bounded by one. Plugging into \cref{thm:JRS} we obtain an adiabatic error at $s=s_-$ of
\begin{align}\label{eq:adiabatic-error}
A(s_-)&=\frac{k}{T}\left[g^{p-2}(0)+g^{p-2}(f(s_-))+k\int_0^{s_-}g^{2p-3}(f(\tau))\left[\dot{g}(f(\tau))p+7\right]d\tau \right]\nonumber \\
&=\frac{k}{T}\left[1+g^{p-2}(\tilde{s}_-)+\int_0^{\tilde{s}_-}g^{p-3}(u)\left[\dot{g}(u)p+7\right]du \right]\nonumber \\
&=\frac{k}{T}\left[1+g^{p-2}(\tilde{s}_-)+\frac{p}{2-p}\left(g^{p-2}(\tilde{s}_-)-1\right)+7\int_0^{\tilde{s}_-}g^{p-3}(u)du \right]
\end{align}
where in the second line we use a change of variables $u=f(\tau)$ and let $\tilde{s}_-=f(s_-)$. 
To evaluate the scaling of \cref{eq:adiabatic-error}, note that $k=\mathcal{O}(g^{1-p}(1/2))=\mathcal{O}(N^{(p-1)/2})$, as the function $g$ is minimized with argument $1/2$ and $p>1$. Furthermore, assume that $\int_0^{\tilde{s}_-}g^{p-3}(u)du = \mathcal{O}\left(g^{p-2}(\tilde{s}_-)\right)$. While this assumption can be proved directly, it is simpler to demonstrate it is true once we determine what $\tilde{s}_-$ must be. 

Let us suppose that $T=\Theta(g^{x-1}(1/2))=\Theta(N^{(1-x)/2})$ for some (possibly $p$-dependent) constant $0<x<1-\frac{p}{2}$. Note, $x>0$ implies better than $\Theta(N^{1/2})$ runtime for Grover search. We will show, via our bounds, that this will yield an $o(1)$ success probability implying that our results give a tight lower bound on Grover search, as $\Theta(N^{1/2})$ scaling is optimal. With such a $T$, $\frac{k}{T}=\mathcal{O}(N^{(x+p)/2-1})=o(1)$.

We require the adiabatic error $A(s_-)=o(1)$, so it must hold that $g^{p-2}(\tilde{s}_-)=o(N^{1-(x+p)/2)})$, or, equivalently, that $g(\tilde{s}_-)=\omega\left(N^{\frac{x+p-2}{2(2-p)}}\right)=\omega\left(N^{\frac{x}{2(2-p)}-\frac{1}{2}}\right)$. This can be accomplished by picking $\tilde{s}_-$ such that $g(\tilde{s}_-)=N^{(\alpha-1)/2}$ for some constant $1>\alpha>\frac{x}{2-p}$. Inverting this and using \cref{eq:gap} implies that the $\tilde{s}$ satisfying this condition are given by
\begin{equation}\label{eq:tildes}
\tilde{s}_\pm=\frac{1\pm\sqrt{1-\frac{N-N^\alpha}{N-1}}}{2}.
\end{equation}
Given this explicit expression for $\tilde{s}_-$ one can show that, indeed, $\int_0^{\tilde{s}_-} g^{p-3}(u) du=\mathcal{O}(g^{p-2}(\tilde{s}_-))$. In particular,
\begin{align}\label{eq:integralofgap}
\int_0^{\tilde{s}_-} g^{p-3}(u) du &= \int_{-\frac{1}{2}}^{\tilde{s}_--\frac{1}{2}} dv \left(\frac{1}{N}+4\frac{N-1}{N}v^2\right)^{\frac{p-3}{2}} \nonumber \\
&=\left[vN^{\frac{3-p}{2}}\,_2F_1\left(\frac{1}{2},\frac{3-p}{2};\frac{3}{2}, -4(N-1)v^2\right)\right]_{-\frac{1}{2}}^{\tilde{s}_--\frac{1}{2}}
\end{align}
where, in the second line we used the change of variables $v=\tilde{s}_--\frac{1}{2}$ and $_2F_1(a,b;c;x)$ is the Gauss hypergeometric function. From \cref{eq:tildes}, $\tilde{s}_--\frac{1}{2}=-\frac{1}{2}N^{\frac{\alpha-1}{2}}+\mathcal{O}(N^{-\frac{1}{2}})$, so, asymptotically, we must only consider the $s=-\frac{1}{2}$ boundary of \cref{eq:integralofgap}. That is, asymptotically in $N$, we have that
\begin{align}
\int_0^{\tilde{s}_-} g^{p-3}(u) du &=\mathcal{O}\left(N^{\frac{3-p}{2}}\,_2F_1\left(\frac{1}{2},\frac{3-p}{2};\frac{3}{2}, 1-N\right)\right)\nonumber \\
&= \mathcal{O}\left(N^{\frac{3-p}{2}}\left[\left(\frac{\sqrt{\pi}\,\gamma\left(1-\frac{p}{2}\right)}{2\, \gamma\left(\frac{3-p}{2}\right)}\right)\frac{1}{\sqrt{N}} +\mathcal{O}\left(N^{\frac{p-3}{2}}\right)\right]\right)=\mathcal{O}\left(N^{1-\frac{p}{2}}\right)
\end{align}
where the second line comes from the asymptotics of $_2F_1\left(\frac{1}{2},b;\frac{3}{2};x\right)$ for $x\rightarrow -\infty$ and $\frac{1}{2}<b<1$. As we picked $\tilde{s}_-$ such that $g^{p-2}(\tilde{s}_-)=N^{\frac{(\alpha-1)(p-2)}{2}}=\Omega\left(N^{1-\frac{p}{2}}\right)$ (given that $0<\alpha<1$), we have shown the desired fact that, for such a $\tilde{s}_-$, $\int_0^{\tilde{s}_-} g^{p-3}(u) du=\mathcal{O}(g^{p-2}(\tilde{s}_-))$. Therefore, if we pick such a $\tilde{s}_-$, the standard adiabatic theorem ensures that at $s= s_-$ the state is $\ket{\lambda(s_-)}+o(1)$.

We now estimate the leakage error.  While it is not necessary to obtain a bound that demonstrates the high probability of missing our target, this approach enables us to show that, on the appropriate timescale, no constant is large enough to guarantee convergence. In fact, the probability of hitting our target decreases with scale. Between $s_-$ and $s_+$, as previously stated, we only have to worry about leakage error, so application of our theorem is relatively straightforward. In fact, relatively drastic additional simplifications are possible. It follows directly from \cref{eq:Hgrover} that $\Delta=\frac{f(s)(N-1)}{N}(\ket{m}\bra{\perp}+\ket{\perp}\bra{m})$, $\mathcal{M}=\ket{m}\bra{m}+\ket{\perp}\bra{\perp}$, and $\Delta^\perp=0$. Therefore, looking at \cref{eq:MdeltaMterm} in \cref{lemma:cheegerperturb} we see that every term in the second line except the $\mathcal{M}\Delta\mathcal{M}$ term vanishes. By \cref{lem:expected_norm}, $\norm{\mathcal{M}\Delta\mathcal{M}}\leq h$, and, as we are applying our method from ${s}_-$ to ${s}_+$, we are left with a leakage error bounded by $\sqrt{T \int_{s_-}^{s_+} h(\tilde{s}) ds }$. It holds that 
\begin{align}
\int_{s_-}^{s_+} h(\tilde{s}) ds= \int_{\tilde{s}_-}^{\tilde{s}_+}\mathrm{d}\tilde{s}\frac{\mathrm{d}s}{\mathrm{d}\tilde{s}}h(\tilde{s}) \nonumber\leq \left[\int_{\tilde{s}_-}^{\tilde{s}_+}\mathrm{d}\tilde{s}\frac{\mathrm{d}s}{\mathrm{d}\tilde{s}}\right]&\left[\int_{\tilde{s}_-}^{\tilde{s}_+}\mathrm{d}\tilde{s}h(\tilde{s})\right]    = \left(s_+-s_-\right)\int_{\tilde{s}_-}^{\tilde{s}_+}\mathrm{d}\tilde{s}h(\tilde{s})\nonumber
    \\&\quad\leq \left(s_+-s_-\right)\int_{\tilde{s}_-}^{\tilde{s}_+}\mathrm{d}\tilde{s}g(\tilde{s}) \nonumber
    =\mathcal{O}\left(N^{\alpha-1}\right),
\end{align}
where, in the next-to-last line, we use that one can bound $h(\tilde{s})\leq g(\tilde{s})$. Using the fact that dynamics are restricted to a two dimensional subspace this is relatively straightforward to demonstrate by explicit computation of the Cheeger constant, but this is shown to be a more general inequality for a broader class of problems in Theorem 7 of Ref.~\cite{jarret2018quantum}. This is tighter than the more general inequality in \cref{eq:stoq_bnd}. In the last line, we use that, trivially, $(s_+-s_-)=\mathcal{O}(1)$, $g(\tilde{s})=\mathcal{O}(g(\tilde{s}_-))=\mathcal{O}(N^{\frac{\alpha-1}{2}})$ for $\tilde{s}\in[\tilde{s}_-,\tilde{s}_+]$, and $\tilde{s}_+-\tilde{s}_-=\sqrt{1-\frac{N-N^\alpha}{N_1}}=\mathcal{O}(N^{\frac{\alpha-1}{2}})$ to obtain an upper bound.

Thus, the leakage error if $T=\Theta(N^{(1-x)/2})$ is $\mathcal{O}\Big(N^{\frac{1}{2}\left(\alpha-\frac{x+1}{2}\right)}\Big).$ This is $o(1)$ if $\alpha< \frac{x+1}{2}.$ From the bounds on $\alpha$ in its definition this holds if and only if $\frac{x}{2-p}< \frac{x+1}{2}$ $\implies$ $x< \frac{2-p}{p}.$ For arbitrary $0<x<1,$ there exists a $1<p<2$ such that this is true.

Therefore, given these parameter choices, which allow us to have rigorous control over the error, at ${s}_+$ the state is still $\ket{\lambda(s_-)} + o(1)$. However, one can show (and we do in the supplemental material, for completeness) that $|\bra{\lambda(s_-)}\lambda(s_+)\rangle|=o(1)$, so we cannot apply the standard adiabatic theorem for $s>s_+$ to prepare the ground state of $H(0)$. That is, for $T=\mathcal{O}(N^{\frac{1-x}{2}})$ when $x>0$, we have no way to rigorously guarantee a successful annealing procedure. This is consistent with the fact that the Grover speedup is optimal.

\section{Discussion}\label{s:discussion}
The Grover search example points to a more general qualitative feature of our bounds: they allow us to rigorously understand annealing beyond the adiabatic regime by bounding the effects of leakage between $S$ and $\overline{S}$. However, ensuring that leakage is small means that one can only guarantee a successful annealing schedule if there is sufficient overlap between the ground state prior to the avoided crossing (when the evolution is still adiabatic) and after the avoided crossing (when the evolution becomes adiabatic again). As we have shown, this is not the case for the Grover problem, but there is experimental and numerical evidence showing that, in some other problems, one can move quickly through an avoided crossing and still maintain a constant success probability. (e.g. Refs.~\cite{crosson2021prospects,ebadi2022quantum,cain2023quantum,braida2023anti}.)
We conjecture that this feature of already having sufficient overlap with the final ground state prior to an avoided crossing is, in fact, the only generic route towards faster-than-adiabatic annealing. That is, we suspect that successful annealing beyond the adiabatic regime that does not follow mechanisms similar to that considered here must, in some sense, depend on an uncommonly ``lucky'' problem instance. This conjecture is falsifiable: if, indeed, the success of certain faster-than-adiabatic schedules can be described by our theorem, one could calculate a slower, but still non- (globally) adiabatic timescale where the leakage error becomes large and, consequently, where we expect the success probability should drop. 

Such conjectures about the extent of the conceptual reach of our theorem aside, we have demonstrated that at timescales short of the adiabatic limit, local dynamics can dominate over global dynamics. Thus, quantum annealing beyond the adiabatic regime can, at least in some cases, be understood in terms of these local dynamics.

As a final point, Ref.~\cite{de2022locobatic} considers a disordered (random) Hamiltonian $H(s)$ acting on a finite lattice and demonstrate what they dub a ``locobatic theorem''---namely, they show that, under certain assumptions, the evolution of an eigenvector $\psi(s)$ of $H(s)$ stays close to the instantaneous eigenvector of the restriction of H(s) to a local region of support of the initial eigenvector $\psi(0)$. Here, the presence of disorder in the corresponding quantum system leads to spectral localization (e.g. in the Anderson model). Importantly, the assumptions that lead to our results do not depend on disorder, enabling their application to adiabatic quantum computation. Despite this, the qualitative features of the results in Ref.~\cite{de2022locobatic} are similar to our own, further highlighting the importance of localized adiabatic behavior for understanding globally non-adiabatic dynamics.

\subsection{Open questions and future work}
There are a number of extensions and improvements we desire out of our analysis in order to yield a theorem of maximum utility in the intended situations. Although we have shown that, in fact, the theorem does provide us with lower bounds on computational problems of interest, we would also like to be able to completely analyze behavior of the situations alluded to in the previous discussion. Some obvious next steps include:

\paragraph{A simplified proof.} We anticipate that a simplified proof of some of the claims of this paper can be obtained by working directly with unitaries and some additional perturbation bounds. Avoiding direct bounds through the resolvent formalism may simplify matters, but will require appropriate tricks in order to maintain or improve the tightness of our bounds.

\paragraph{Improving dependence upon $\Gamma_{\overline{S}}$ and $\kappa$.} We anticipate that the dependency of our bounds upon $\Gamma_{\overline{S}}$ are an artifact of analysis. We expect that any gap dependency can be ultimately (or at least largely) restricted to the local evolution within the set $S$. The dependency upon $\kappa$ seems harder to eliminate, although it also appears to be an artifact of analysis. Even if this term cannot be eliminated, reducing its importance would improve the strength and simplicity of our results.

\paragraph{Applications.} In this paper we offered one compelling application as an example, a physical speed limit on adiabatic Grover search. We expect that we can also turn our analysis to a number of well-known no-go results and demonstrate that adiabatic or annealing versions of these protocols will also fail for physical reasons.

\paragraph{Higher order Cheeger inequalities.} A large body of recent work on Cheeger inequalities has shown that one can define higher order Cheeger inequalities to obtain bounds on different eigenvalues within the spectrum \cite{lee2014multi}. Such inequalities have not been generalized to the case of spectral gaps of Hamiltonians and the nature of such a generalization is not obvious. Nonetheless, such inequalities might allow us better control and the ability to avoid utilizing a ``global gap'' between the low and high energy subspaces.

\paragraph{Time-dependent Cheeger ratios and cuts.} To the authors' knowledge, no work has been done on a time-dependent version of $h$ where the subset $S$ is also time-dependent. This is because the definition of $h$ can have sharp transitions over which set achieves the infimum as $H$ varies. We expect that stronger inequalities for deriving upper and lower bounds will require that $S$ be allowed to vary in time, such that the set $S$ can be implicitly defined based on properties of the ground state. (For instance, the Cheeger constant minimizes $h_S$ over all choices of $S$ and the $S$ that achieves the infimum will depend on time and may not be unique.) Understanding a version of our theorem in terms of appropriate time-dependent quantities that can track a changing $S$ will allow us to better understand the dynamics of how the ``bulk'' of a distribution moves throughout Hilbert space under continuous variations. Such adaptations would also seemingly greatly simplify proofs of lower bounds like in the previous section.

\paragraph{Acknowledgements.} We thank Michael J. O'Hara and Andrew Glaudell for useful discussions. This material is based upon work supported
by the Army Research Office under Contract No. [W911NF-20-C-0025]. J.B.~was supported in part by the U.S.~Department of Energy, Office of Science, Office of Advanced Scientific Computing Research, Department of Energy Computational Science Graduate Fellowship (award No.~DE-SC0019323) and in part by the National Science Foundation under Grant No. NSF PHY-1748958, the Heising-Simons Foundation, and the Simons Foundation (216179, LB) and in part by the Harvard Quantum Initiative. J.B. and T.C.M. acknowledge funding by the U.S. Department of Energy (DOE) ASCR Accelerated Research in Quantum Computing program (award No.~DE-SC0020312), DoE QSA, NSF QLCI (award No.~OMA-2120757), DoE ASCR Quantum Testbed Pathfinder program (award No.~DE-SC0019040), NSF PFCQC program, AFOSR, ARO MURI, AFOSR MURI, and DARPA SAVaNT ADVENT.

\bibliography{main.bib}
\bibliographystyle{siam}

\setcounter{secnumdepth}{1}
\setcounter{section}{0}
\renewcommand{\thesection}{S\arabic{section}}
\setcounter{theorem}{0}
\renewcommand{\thetheorem}{S\arabic{theorem}}
\setcounter{lemma}{0}
\renewcommand{\thelemma}{S\arabic{lemma}}
\setcounter{equation}{0}
\renewcommand{\theequation}{S\arabic{equation}}
\setcounter{table}{0}
\renewcommand{\thetable}{S\arabic{table}}
\setcounter{figure}{0}
\renewcommand{\thefigure}{S\arabic{figure}}

\center{\Large\textbf{Supplemental Material}}
\vspace{1em}

\raggedright
In this supplemental material, we provide technical details for a few proofs that serve important, but ancillary, roles in the derivations of the main text. In particular, in \cref{s:resolvent-proofs} we provide some additional lemmas and ultimately prove Lemma 10 of the main text, which is used to bound the difference between $U_\perp'(1)$ and $V_\mathrm{ad}(1)$ in the main text. In \cref{app:gapbnd} we provide a derivation of the fact that $|\bra{\lambda(s_-)}\lambda(s_+)\rangle|=o(1)$ in the Grover search example in Section 7b of the main text.

\section{Additional Lemmas for the Perturbed Adiabatic Theorem}\label{s:resolvent-proofs}
Recall that the quantity $F$ is defined in Eq. (5.17) of the main text as
\begin{equation}\label{def:F-s}
   F:=R_\perp \dot{P}_{\mu}P_{\mu}+P_{\mu}\dot{P}_{\mu}R_\perp.
\end{equation}
and, from Eq. (5.14) in the main text,
\begin{equation}\label{fact:ResolventIdentity-s}
    L := (I+R\Cut^\perp).
\end{equation}
We refer the reader to the main text for the definitions of other quantities. Armed with this definition, we prove a few additional lemmas.

\begin{prop}
    Let $F$ be defined as above in \cref{def:F}. It holds that 
    \[F = L^{-1}R\dot{P}_\mu P_\mu + P_\mu\dot{P}_\mu L^{-1}R\]
    and 
    \[\dot{F}P_\mu =-L^{-1}\left(\dot\Cut^\perp R+\Cut^\perp\dot{R}\right)L^{-1}R\dot{P}_\mu P_\mu + L^{-1}\dot{R}\dot{P}_\mu P_\mu + L^{-1}{R}\ddot{P}_\mu P_\mu  - P_\mu\dot{P}_\mu L^{-1}{R}\dot{P}_\mu.\]
\end{prop}
\begin{proof}
    The first equality follows directly from the definition of $F$ and \cref{fact:ResolventIdentity-s}. The second comes immediately from differentiating the first and using the fact that $RP_\mu=0$, eliminating the second term.
\end{proof}

\begin{lemma} Let $H',\Cut^\perp, F$ and $P_{\mu}$ be as defined above. Then,
    \[
    [H'+\Cut^\perp, F] =[{\dot{P}}_{\mu},P_{\mu}].
    \]\label{lem:HFPmuPmubar}
\end{lemma}
\begin{proof}
    Simply plugging in $F$ from its definition in \cref{def:F}, noting that $[\mu I, F] = 0$, and using that $(H'+\Cut^\perp-\mu I )R_\perp = I-P_\mu-P_{\overline{\mu}},$ we can proceed by direct computation:
    \begin{align*}    
    [H'+&\Cut^\perp, F] = [H'+\Cut^\perp, R_{\perp}\dot{P}_{\mu}P_{\mu} + P_{\mu}\dot{P}_{\mu}R_{\perp}]\\
    &= (H'+\Cut^\perp )R_{\perp}\dot{P}_{\mu}P_{\mu} + (H'+\Cut^\perp){P_{\mu}}\dot{P_{\mu}}R_{\perp} - R_{\perp}\dot{P_{\mu}}P_{\mu} (H'+\Cut^\perp) - P_\mu\dot{P}_\mu R_{\perp}(H'+\Cut^\perp)\\
    &= (H'+\Cut^\perp - \mu I)R_\perp \dot{P}_{\mu}P_{\mu}-P_{\mu}\dot{P}_{\mu}R_\perp(H'+\Cut^\perp-\mu I)\\
    &= {P_{\mu}}^\perp \dot{P}_{\mu}P_{\mu} - P_{\mu}\dot{P}_{\mu}P_{\mu}^\perp \\
    &= \dot{P}_{\mu}P_{\mu}-P_{\mu}\dot{P}_\mu= [\dot{P}_{\mu},P_{\mu}].
    \end{align*}
\end{proof}
Lemma \ref{lem:HFPmuPmubar} allows us to prove the following result, which is stated as Lemma 10 in the main text.

\begin{lemma} \label{lemma:UperpPmuPmubarUperpsupp}
Let $U_\perp',$ ${P}_{\mu},$ and $F$ be as defined above. Then, 
    \[(U_\perp')^{\dagger}[\dot{{P}}_{\mu},{P}_{\mu}]U_\perp' = \frac{1}{iT}\left(\dot{\tilde{F}}-(U_\perp')^{\dagger}\dot{F}U_\perp'\right),\] where $\tilde{F}:=(U_\perp')^{\dagger} FU_\perp'.$
\end{lemma}
\begin{proof}
    Plugging in the result of \cref{lem:HFPmuPmubar} we get that
    \begin{align*}
        (U_\perp')^{\dagger}[\dot{{P}}_{\mu},{P}_{\mu}]U_\perp' &= (U_\perp')^{\dagger}(H'+\Delta^\perp) F U_\perp' -  (U_\perp')^{\dagger} F(H'+\Delta^\perp)U_\perp' \\
        &=\frac{1}{iT}\left((\dot{U}'_\perp)^\dagger F U_\perp' +(U_\perp')^{\dagger} F \dot{U}_\perp'\right)\\
        &=\frac{1}{iT}\left(\dot{\tilde{F}}-(U_\perp')^{\dagger}\dot{F}U_\perp'\right),
    \end{align*}
    where in the second line we used Eq.~(2.4) of the main text and $\tilde{F}:=(U_\perp')^{\dagger} FU_\perp'$ as in the lemma statement.
\end{proof}

\section{Details for the Grover Search Example}\label{app:gapbnd}
In the main text, we claimed that  $|\bra{\lambda(s_-)}\lambda(s_+)\rangle|=o(1)$. Here, we demonstrate this fact. To begin, we use that the dynamics are confined to a two dimensional subspace, allowing us to diagonalize the Grover Hamiltonian
\begin{align}\label{eq:Hgroversupp}
&H(s)=(1-f(s))H_1+f(s)H_0\nonumber \\
&=I-\left(1-\frac{f(s)(N-1)}{N}\right)\ket{m}\bra{m}-\frac{f(s)(N-1)}{N}\ket{\perp}\bra{\perp} -\frac{f(s)\sqrt{N-1}}{N}\left(\ket{\perp}\bra{m}+\ket{m}\bra{\perp}\right),
\end{align}
in the $\{\ket{m},\ket{\perp}\}$ basis.  Recall this Hamiltonian has an eigenvalue gap given by Eq.~(7.1) of the main text as
\begin{equation}\label{eq:gapsup}
g(\tilde s):=\sqrt{1-4\cdot \frac{N-1}{N}\tilde{s}(1-\tilde{s})}.
\end{equation}
In particular, in this subspace, the Hamiltonian can be written as 
\begin{equation}
H(\tilde{s})=
\begin{pmatrix}
\tilde{s}\left(1-\frac{1}{N}\right) & -\tilde{s}\frac{\sqrt{N-1}}{N}\\-\tilde{s}\frac{\sqrt{N-1}}{N}&\left(1-\tilde{s}\frac{N-1}{N}\right)\end{pmatrix}.
\end{equation}
Finding the eigenvalues and eigenvectors of a  symmetric $2\times 2$ matrix  is straightforward. Consequently, we have that the ground state is proportional to 
\begin{equation}
\left(1-2\tilde{s}\left(1-\frac{1}{N}\right) + \sqrt{1-4\tilde{s}(1-\tilde{s})\left(1-\frac{1}{N}\right)}\right)\ket{m}+2\tilde{s}\sqrt{\frac{1}{N}\left(1-\frac{1}{N}\right)}\ket\perp=:\alpha(\tilde{s})\ket{m}+\beta(\tilde{s})\ket\perp,
\end{equation}
where the coefficients $\alpha(\tilde{s})$, $\beta(\tilde{s})$ are defined implicitly. Straightforward algebra gives a normalized ground state 
\begin{align}\label{eq:state}
\ket{\lambda(s)} &= \frac{\left(1-2\tilde{s}\left(1-\frac{1}{N}\right)+\sqrt{1-4\tilde{s}(1-\tilde{s})\left[1-\frac{1}{N}\right]}\right)\ket m + 2\tilde{s}\sqrt{\frac{1}{N}\left(1-\frac{1}{N}\right)}\ket\perp}{\sqrt{2\left[1-4\tilde{s}(1-\tilde{s})\left(1-\frac{1}{N}\right)+\left(1-2\left[1-\frac{1}{N}\right]\tilde{s}\right)\sqrt{1-4\tilde{s}(1-\tilde{s})\left[1-\frac{1}{N}\right]}\right]}} \nonumber \\
&= \frac{\left(1-2\tilde{s}\left(1-\frac{1}{N}\right)+g(\tilde{s})\right)\ket m + 2\tilde{s}\sqrt{\frac{1}{N}\left(1-\frac{1}{N}\right)}\ket\perp}{\sqrt{2g(\tilde{s})\left[g(\tilde{s})+\left(1-2\left[1-\frac{1}{N}\right]\tilde{s}\right)\right]}} \nonumber \\
&= \frac{\left(1-2\tilde{s}\left(1-\frac{1}{N}\right)+g(\tilde{s})\right)\ket m + 2\tilde{s}\sqrt{\frac{1}{N}\left(1-\frac{1}{N}\right)}\ket\perp}{\mathcal{N}(\tilde{s})}
\end{align}
where $\mathcal{N}(\tilde{s})=\sqrt{\alpha(\tilde{s})^2+\beta(\tilde{s})^2}$ is the implicitly defined normalization factor on the state and $g(\tilde{s})$ is eigenvalue gap given in \cref{eq:gapsup}. 
Now
\begin{equation}
\bra{\lambda(s_-)}\lambda(s_+)\rangle=\frac{\alpha(\tilde{s}_-)\alpha(\tilde{s}_+)+\beta(\tilde{s}_-)\beta(\tilde{s}_+)}{\mathcal{N}(\tilde{s}_-)\mathcal{N}(\tilde{s}_+)}=\frac{\alpha(\tilde{s}_-)\alpha(1-\tilde{s}_-)+\beta(\tilde{s}_-)\beta(1-\tilde{s}_-)}{\mathcal{N}(\tilde{s}_-)\mathcal{N}(\tilde{s}_+)}.
\end{equation}
With a bit of algebra one obtains
\begin{align}
    \alpha(\tilde{s})\alpha(1-\tilde{s})&= 
    -4\tilde{s}(1-\tilde{s})\frac{1}{N}\left(1-\frac{1}{N}\right) +\frac{2}{N}\left(1+g(\tilde{s})\right),
\end{align}
and
\begin{equation}
\beta(\tilde{s})\beta(1-\tilde{s}) = 4\tilde{s}(1-\tilde{s})\frac{1}{N}\left(1-\frac{1}{N}\right),
\end{equation}
and, thus, 
\begin{align}
    \alpha(\tilde{s})\alpha(1-\tilde{s})+\beta(\tilde{s})\beta(1-\tilde{s})= \frac{2}{N} \left(1+g(\tilde{s})\right)
\end{align}

Using that $\tilde{s}=\tilde{s}_- = \frac{1-\sqrt{\frac{N^\alpha-1}{N-1}}}{2}$ and $g(\tilde{s}_-)=N^{\frac{\alpha-1}{2}}$, as in the main text, we have that 
\begin{align}
    \mathcal{N}^2(\tilde{s}_-) &= 2N^{\frac{\alpha-1}{2}}\left[1-\left(1-\frac{1}{N}\right)\left(1-\sqrt{\frac{N^\alpha-1}{N-1}}\right)+N^{\frac{\alpha-1}{2}}\right] \nonumber\\
    &= 2N^{\frac{\alpha-1}{2}}\left[1-\left(1-\frac{1}{N}-\frac{1}{N}\sqrt{(N-1)(N^\alpha-1)}\right)+N^{\frac{\alpha-1}{2}}\right] \nonumber\\
    &= 2N^{\frac{\alpha-1}{2}}\left[N^{\frac{\alpha-1}{2}} + \frac{1}{N}\left(1+\sqrt{N^{\alpha+1}-N-N^\alpha+1}\right)\right]\nonumber \\
    &= \Omega\left(N^{\alpha-1}\right)
\end{align}
Similarly, we have that
\begin{align}
    \mathcal{N}^2(\tilde{s}_+) &= 2N^{\frac{\alpha-1}{2}}\left[1-\left(1-\frac{1}{N}\right)\left(1+\sqrt{\frac{N^\alpha-1}{N-1}}\right)+N^{\frac{\alpha-1}{2}}\right]\nonumber\\
    &= 2N^{\frac{\alpha-1}{2}}\left[1-1+\frac{1}{N}-\frac{1}{N}\sqrt{(N-1)(N^\alpha-1)}+N^{\frac{\alpha-1}{2}}\right]\nonumber\\
    &= 2N^{\alpha-1}\left[1-\sqrt{1-N^{-1}-N^{-\alpha}+N^{-(\alpha+1)}}\right]\nonumber\\
    &= 2N^{\alpha-1}\left[1-\left(1+\frac{-N^{-1}-N^{-\alpha}+N^{-\alpha-1}}{2}+\mathcal{O}\left(N^{-2\alpha}\right)\right)\right]\nonumber\\
    &= N^{\alpha-1}\left[N^{-\alpha}+N^{-1}+\mathcal{O}(N^{-2\alpha})\right]\nonumber\\
    &= \Omega(N^{-1}).
\end{align}
Thus, $\mathcal{N}(\tilde{s}_-)\mathcal{N}(\tilde{s}_+)=\Omega\left(N^{\frac{\alpha}{2}-1}\right).$ Thus, 
\begin{equation}
\langle\lambda(s_-)\vert \lambda(s_+)\rangle=\frac{2\left(1+N^{\frac{\alpha-1}{2}}\right)}{N\Omega\left(N^{\frac{\alpha}{2}-1}\right)}= \frac{\mathcal{O}(1)}{\Omega(N^{\alpha/2})}=\mathcal{O}\left(N^{-\alpha/2}\right)=o(1).
\end{equation}

\end{document}